\documentclass[reprint,aps,prl]{revtex4-2}
\pdfoutput=1
\usepackage{amsthm,amssymb,amsmath}
\usepackage{enumerate}
\usepackage{hyperref}
\usepackage{graphicx}

\newcommand{\prlsection}[1]{\paragraph*{#1.---}\hspace{-1em}}

\makeatletter
\def\maketitle{
\@author@finish
\title@column\titleblock@produce
\suppressfloats[t]}
\makeatother

\usepackage{tikz}
\usetikzlibrary{arrows}

\tikzset{
box/.style={rectangle,draw,minimum height=1cm,minimum width=2cm,fill=gray!15},
invis/.style={rectangle,minimum height=1cm,minimum width=2cm}
}

\newcommand{\ie}{i.e.~}
\newcommand{\eg}{e.g.~}

\theoremstyle{plain}
\newtheorem{theorem}             {Theorem}

\newtheorem{lemma}      [theorem]{Lemma}

\newtheorem*{theorem*}    {Theorem}
\newtheorem*{proposition*}{Proposition}
\newtheorem*{lemma*}      {Lemma}
\newtheorem*{corollary*}  {Corollary}
\newtheorem*{conjecture*} {Conjecture}

\theoremstyle{definition}
\newtheorem{definition}[theorem]{Definition}

\newtheorem*{definition*}{Definition}
\newtheorem*{example*}   {Example}

\theoremstyle{remark}
\newtheorem{remark}[theorem]{Remark}

\usepackage{colonequals}
\newcommand{\defeq}{\colonequals} 
\newcommand{\tuple}[1]{\mathopen{\langle}#1\mathclose{\rangle}} 
\newcommand{\setdef}  [2]{\left\{#1 \mid #2\right\}}             % { f(x) | p(x) }
\newcommand{\enset}   [1]{\mathopen{ \{ }#1\mathclose{ \} }} % {a,b,...z}

\newcommand{\fdef}    [3]{#1 \colon\colon #2 \longmapsto     #3}

\newcommand{\MP}{\mathsf{MP}}

\newcommand{\NCF}{\mathsf{NCF}}

\newcommand{\x}{\bar{x}}
\newcommand{\y}{\bar{y}}

\newcommand{\simulates}{\rightsquigarrow}

\newcommand{\X}{\mathcal{X}}

\newcommand{\Ev}{\mathcal{E}}

\newcommand{\scen}[1]{\tuple{X_#1,\Sigma_#1,O_#1}}

\begin{document}

\title{Neither Contextuality nor Nonlocality Admits Catalysts}
\author{Martti Karvonen}
\affiliation{University of Ottawa, Canada}

\begin{abstract}
We show that the resource theory of contextuality does not admit catalysts, i.e., there are no correlations that can enable an otherwise impossible resource conversion and still be recovered afterward. As a corollary, we observe that the same holds for nonlocality. As entanglement allows for catalysts, this adds a further example to the list of ``anomalies of entanglement,'' showing that nonlocality and entanglement behave differently as resources. We also show that catalysis remains impossible even if, instead of classical randomness, we allow some more powerful behaviors to be used freely in the free transformations of the resource theory.
\end{abstract}

\maketitle

\prlsection{Introduction}Contextuality~\cite{ks,Contextualityreview} and nonlocality~\cite{bell1966,bell-review} play a prominent role in a wide variety of applications of quantum mechanics, with nonlocality being used, for example, in quantum key-distribution~\cite{scarani2009security}, certified randomness~\cite{acin2016certified} and randomness expansion~\cite{fehr2013security}. Similarly, contextuality powers quantum computation in some computational models~\cite{Anders2009,raussendorf2013contextuality,howard2014contextuality,bermejo2017contextuality,raussendorf2017contextuality,abramsky2017contextual} and even increases expressive power in quantum machine learning~\cite{quantumML}. Consequently, it is vital to understand how nonlocality and contextuality behave as resources. 

In this Letter, we show that neither contextuality nor nonlocality admits catalysts: that is, there are no correlations that can be used to enable an otherwise impossible conversion between correlations and still be recovered afterward. Slightly more precisely, let us write $d,e,f\dots$ for various correlations (whether classical or not), $d\otimes e$ for having independent instances of $d$ and $e$, and  $d\simulates e$ (read as ``$d$ simulates $e$'') for the existence of a conversion $d\to e$. Then our results state that, in suitably formalized resource theories of contextuality and nonlocality, whenever $d\otimes e\simulates d\otimes f$, then $e\simulates f$ already. This gives a strong indication that contextuality (and nonlocality) are resources that \emph{get spent} when you use them: there is no way of using a correlation $d$ to achieve a task you could not do otherwise while keeping $d$ intact.  As entanglement theory famously allows for catalysts~\cite{catalysisofentanglement}, this can be seen as yet another  ``anomaly of nonlocality''~\cite{anomalyofnon-locality} and thus further testament to the fact that nonlocality and entanglement are different resources. 

We prove our results by working in precisely defined resource theories of contextuality and nonlocality. These are not strictly speaking quantum resource theories~\cite{chitambar2019resource}, but resource theories in a more general sense~\cite{coecke2016mathematical,fritz2017resource}, as we allow resources such as  Popescu-Rohrlich boxes (PR boxes)~\cite{popescu1994quantum} that are not quantum realizable. The kinds of conversions between correlations we have in mind capture the intuitive idea of using one system to \emph{simulate} another one, and have been studied in earlier literature~\cite{barrettetal2005,barrettpironio2005,allcock2009closedsets,jonesmasanes2005interconversions,dupuis2007nouniversalbox,vicente:nonlocality,forsterwolf2011bipartite,gallego2012operational,Gallego:wirings}. These roughly correspond to the local operations and shared randomness paradigm (LOSR) or to wirings and prior-to-input classical communication, depending on the precise definitions of these terms. However, existing formalizations of these in the literature are often limited to the bipartite or tripartite settings, and at times overlook some technical issues resulting in nonconvex sets of transformations~\cite[Appendix]{Wolfeetal:quantifyingbell}.

More importantly, existing formalizations of the resource theory of nonlocality tend to focus on the case where each party has a discrete set of measurements of which they can perform at most one. However, this is false even in relatively simple situations: for instance, if Alice shares one PR box with Bob and one with Charlie, then she has four measurements available but is not restricted to only one measurement as she can choose a measurement for each of her boxes. In particular, she might first measure one of the boxes and use the outcome (and possible auxiliary randomness) to choose what to measure next. 

To overcome such issues, we work in the general approach to contextuality initiated in~\cite{ab} and later extended in~\cite{comonadicview} to capture building some correlations from others using such \emph{probabilistic} and \emph{adaptive} means, resulting in a resource theory for contextuality. We then obtain the resource theory of nonlocality from this via a general mathematical construction used in~\cite{crypto} that builds a resource theory of $n$-partite resources from a given resource theory. Working with such generality clarifies the relationship between resource theories of contextuality and nonlocality and captures the kinds of interconversions studied in earlier literature in the exact, single-shot regime~\footnote{This is in contrast to asymptotic questions~\cite[V.B]{chitambar2019resource} or distillation tasks~\cite{Brunner:distillation} where one tries to approximate a target resource using increasing numbers of copies of the starting resource. It is unclear if the notion of a catalyst makes sense in such settings, which is why we work with single-shot convertibility.} in a precise yet tractable manner.

We believe that our result could be phrased and proved in terms other approaches to contextuality~\cite{spekkens2005contextuality,ehtibar2014contextuality,csw2014graphtheoretic,acin2015combinatorial,sander2015effect}, as long as one formalizes such adaptive measurement protocols and transformations between correlations within them. However, the current proof strategy no longer applies if one works with axiomatically defined transformations as in~\cite{joshi2011nobroadcasting}, \ie with abstract functions between sets of correlations satisfying properties such as preservation of locality and of convex combinations. This is because such abstract functions might not arise from operationally defined protocols that one might implement physically. Indeed, operational transformations form a proper subset of axiomatically defined ones for the resource theories of entanglement~\cite{Bennetetal:loccvsaxiomatic} and magic~\cite{Heimendahl:axiomaitvsoperationalmagic}. For contextuality and nonlocality it is not known if the axiomatic and operational resource theories agree, although a characterization of those functions arising from (nonadaptive) operational transformations is given in~\cite[Theorem 44]{closingbell}.

One possible explanation for anomalies of nonlocality, put forward in~\cite{schmid:losr}, is that they stem from using LOSR transformations with nonlocality and local operations and classical communication (LOCC) with entanglement. Indeed, \cite{schmid:losr} shows that many of the anomalies disappear when working with LOSR entanglement. In particular they show, for bipartite pure states, that there are no catalysts for LOSR entanglement. We conjecture that there are no catalysts in general for LOSR entanglement, in which case this anomaly is fully explained by entanglement and nonlocality being measured with LOCC and LOSR transformations respectively. For a contrasting viewpoint, see~\cite{Senguptaetal:Nonlocalityisentanglement} which argues that both entanglement and nonlocality should be measured in terms of LOCC transformations.

\prlsection{No catalysis for contextuality}We begin by briefly reviewing the resource theory of contextuality as defined in~\cite{comonadicview}. To start, we formalize the idea of a ``measurement scenario'' $S$: we imagine a situation where there is a finite set $X_S$ of measurements available, each measurement $x\in X_S$ giving rise to outcomes in some finite set $O_{S,x}$. However, only some measurements might feasible to perform together---other combinations may be ruled out by practical limitations or excluded by physical theory. We collect all jointly compatible measurements into a single set $\Sigma_S$, which we expect to satisfy two natural properties: (1) any measurement $x\in X_S$ induces a compatible set $\{x\}\in\Sigma_S$, and (2) any subset of a compatible set of measurements is compatible. Collecting all this data together results in the measurement scenario $S=\scen{S}$.

Given two scenarios $S$ and $T$, we let $S\otimes T$ denote the scenario that represents having access to $S$ and $T$ in parallel, so that a joint measurement is possible precisely if its components in $S$ and $T$ are possible.

An \emph{empirical model} over a scenario $S$ is given by specifying for each compatible $\sigma\in\Sigma_S$ a joint probability distribution $e_\sigma$ for measurements in $\sigma$. We only consider empirical models for which the behavior of a joint measurement does not depend on what is measured with it, if anything. Thus, whenever $\tau\subset\sigma\in\Sigma_S$,  the distribution $e_\tau$ over $\tau$ can be obtained by marginalizing $e_\sigma$ to $\tau$. We express this generalization of the usual no-signaling conditions as $e_\sigma|_\tau=e_\tau$, so that more generally for a joint distribution $d$ over outcomes of $Y\subset X_S$ and $Z\subset Y$, the expression $d|_Z$ denotes the marginal distribution on outcomes of $Z$. 

If the impossibility of measuring everything together is only a practical limitation, one can contemplate the distribution $d$ that would arise when measuring $X_S$. If $d$ explained the model $e$ we have at hand, we would expect it to satisfy $d|_\sigma=e_\sigma$ for every $\sigma\in\Sigma_S$. If such a distribution exists, we call $e$ ``nonconextual''. If no such distribution exists, we call $e$ ``contextual'', as we have reason to believe that it is infeasible in principle to measure everything together, unless one accepts that the observed joint distribution for a subset $Y\subset X$ depends on what it is measured with i.e., its context.

A simple example of scenario and a contextual empirical model on it, discussed in~\cite{Spekkerstriangle}, is given by three measurements, any two of which are compatible but not all three together. Each measurement takes outcomes in $\{0,1\}$, and whenever two measurements are performed one observes $(0,1)$ and $(1,0)$ with equal probability, with the probability distributions for singletons fixed by this and resulting in observing the two outcomes with equal probability. To see that this is contextual, note that there is no joint outcome for all three measurements that is consistent with the observed marginals. 

More examples can be obtained from scenarios studied in nonlocality, where one typically specifies a scenario by giving the number of parties, the number of measurements available to each of them and the size of the outcome sets, where it is then understood that maximal compatible measurements are given by a choice of a single measurement by each party. This includes for instance the famous Clauser, Home, Shimony, and Holt model~\cite{CHSH} and the PR box~\cite{popescu1994quantum} which goes beyond what is allowed in quantum mechanics, both models arising in the scenario with two parties having access to two dichotomic measurements. 

We now move on to transformations between scenarios and empirical models. We will build up to ``wirings'' that in full generality capture the idea of ``simulating'' simulating an empirical model from another adaptively with the help of noncontextual randomness. These will be the free transformations of our resource theory, but we begin by considering the problem of building one scenario from another. A particularly simple way of building $T$ from $S$, is by declaring that for each measurement $x\in X_T$ of $T$ some measurement $\pi(x)\in X_S$ of $S$ is to be performed instead. Moreover, each outcome $o$ of $\pi(x)$ is to be interpreted as the outcome $\alpha_x(o)$ instead. If for each compatible $\sigma\in\Sigma_S$ the corresponding measurement $\pi(\sigma)$ is jointly compatible in $S$, the pair $\tuple{\pi,\alpha=(\alpha_x)_{x\in X}}$ describes a way of building $T$ from $S$, and we will denote this by writing $\tuple{\pi,\alpha}\colon S\to T$.
 
Given such $\tuple{\pi,\alpha}\colon S\to T$, any empirical model $e:S$ induces a model on $T$ that describes the statistics one would see if one was to observe the statistics given by $e$ and then transform them according to $\tuple{\pi,\alpha}$. We then denote by $\tuple{\pi,\alpha}_*$ the function that pushes empirical models on $S$ forward to empirical models on $T$, so that for $e:S$ the induced model on $T$ is denoted by $\tuple{\pi,\alpha}_* (e)$. The pair $\tuple{\pi,\alpha}$ is defined to be a ``deterministic simulation'' of $e:T$ from $d:S$, denoted by $\tuple{\pi,\alpha}\colon d\to e$, precisely when $\tuple{\pi,\alpha}$ transforms  $d$ to $e$, i.e., if $\tuple{\pi,\alpha}_*(d)=e$. Such a simulation is depicted in Fig.~\ref{fig:procedure} 
 
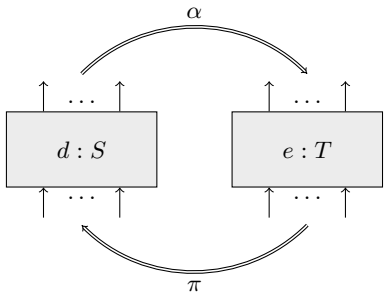
\begin{figure}
\begin{tikzpicture}
\node[box] (a) [label=above:$\ldots$, label=below:$\ldots$] at (0,0){$d:S$};

\node[box] (b) [label=above:$ \ldots $, label=below:$\ldots$] at (3,0) {$e:T$};

\draw [->] ([xshift=.5cm]a.north west) to +(0,.4);
\draw [->] ([xshift=-.5cm]a.north east) to +(0,.4);
\draw [<-] ([xshift=.5cm]a.south west) to +(0,-.4);
\draw [<-] ([xshift=-.5cm]a.south east) to +(0,-.4);

\draw [->] ([xshift=.5cm]b.north west) to +(0,.4);
\draw [->] ([xshift=-.5cm]b.north east) to +(0,.4);
\draw [<-] ([xshift=.5cm]b.south west) to +(0,-.4);
\draw [<-] ([xshift=-.5cm]b.south east) to +(0,-.4);

\draw [->, double,>=implies] ([yshift=.5cm]a.north)  to[out=45,in=135] node[above] {$\alpha$} ([yshift=.5cm]b.north);
\draw [->, double, >=implies] ([yshift=-.5cm]b.south) to[out=-135,in=-45]  node[below] {$\pi$}  ([yshift=-.5cm]a.south);
\end{tikzpicture}\caption{Depiction of a deterministic simulation $d\to e$, with $\pi$ transforming inputs of $T$ to those of $S$ and $\alpha$ transforming outputs of $S$ to outputs of $T$.}\label{fig:procedure}
\end{figure}

There are two ways in which these deterministic simulations are weaker than one would want in a general resource theory of contextuality: first, one might want to allow the usage of auxiliary (noncontextual) randomness, so that the dependence of measurements and their outcomes of $T$ on those of $S$ is stochastic. The usage of auxiliary randomness is captured in~\cite{amaral2017noncontextual} and discussed further in~\cite{amaral2019resource}, and in our terms, could be defined by allowing probabilistic mixtures of transformations. However, this viewpoint leaves out another important generalization: namely  the possibility that a single measurement in $T$ can depend on a joint measurement of $S$ as in~\cite{karvonen2018categories} or more generally on a measurement protocol on $S$ that chooses which joint measurement to perform \emph{adaptively}.  The most general formulation of this idea would allow a measurement $x$ in $T$ to be simulated by a probabilistic and \emph{adaptive} procedure, that first measures (depending on some classical randomness) something in $S$ and then, based on the outcome (and possibly further classical randomness), chooses what to measure next (if anything) and so on. This idea is formalized carefully in~\cite{comonadicview} in two stages, the first one adding adaptivity and the second adding randomness. 

To model adaptivity, one builds from a scenario $S$ a new scenario $\MP(S)$ of (deterministic) ``measurement protocols'' over $S$. A measurement protocol is a procedure that, at any stage, either stops and reports all of the measurement results obtained so far, or, based on previously seen outcomes, performs a measurement in $S$ that is compatible with the previous measurements. The measurements of $\MP(S)$ are given by such protocols over $S$, and a set of measurement protocols is compatible if they can be performed jointly without having to query measurements outside of $\Sigma_S$. Then one can define adaptive (but still deterministic) transformations $S\to T$ between scenarios as deterministic transformations  $\MP(S)\to T$ as in Fig.~\ref{fig:MP}. In~\cite{comonadicview}, we show that the assignment $S\mapsto \MP(S)$ defines a comonad on the category of scenarios: this abstract language is not needed here but can be thought of as guaranteeing that one has a well-behaved way of composing $\MP(S)\to T$ with $\MP(T)\to U$ to obtain a map $\MP(S)\to U$. Intuitively, the composite is obtained by first interpreting each measurement in $U$ as a measurement protocol over $T$, and each measurement in that measurement protocol as an measurement protocol over $S$, and then ``flattening'' the resulting measurement protocol of measurement protocols \{i.e., a measurement of $\MP[\MP((S)]$\} into a measurement protocol over $S$. 

\begin{figure}
\begin{tikzpicture}
\draw[black, fill = gray, fill opacity = 0.5, semithick] (-1.5,-1.25) rectangle (1.5,1.25);
\node at (-1,1) {$\MP$};
\node[box] (a) [label=above:$\ldots$, label=below:$\ldots$] at (0,0){$d: S$};

\node[box] (b) [label=above:$ \ldots $, label=below:$\ldots$] at (3,0) {$e: T$};

\draw [->] ([xshift=.5cm]a.north west) to +(0,.4);
\draw [->] ([xshift=-.5cm]a.north east) to +(0,.4);
\draw [<-] ([xshift=.5cm]a.south west) to +(0,-.4);
\draw [<-] ([xshift=-.5cm]a.south east) to +(0,-.4);

\draw [->] ([xshift=.5cm]b.north west) to +(0,.4);
\draw [->] ([xshift=-.5cm]b.north east) to +(0,.4);
\draw [<-] ([xshift=.5cm]b.south west) to +(0,-.4);
\draw [<-] ([xshift=-.5cm]b.south east) to +(0,-.4);

\draw [->, double,>=implies] ([yshift=.5cm]a.north) to[out=45,in=135]   node[above] {$\alpha$} ([yshift=.5cm]b.north);
\draw [->, double, >=implies] ([yshift=-.5cm]b.south) to[out=-135,in=-45] node[below] {$\pi$} ([yshift=-.5cm]a.south);

\end{tikzpicture}\caption{Depiction of an adaptive simulation $d\to e$}\label{fig:MP}
\end{figure}
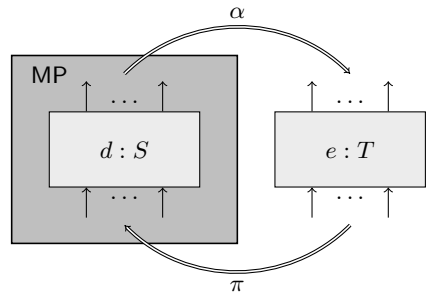

To model (noncontextual) randomness, one then defines a simulation $d\to e$ to be a deterministic simulation $\MP(d\otimes c)\to e$ for some noncontextual model $c$ as in Fig.~\ref{fig:MPandrandomness}. Again, such simulations compose in a well-behaved manner, i.e. they form a category~\footnote{the abstract explanation is that the $\MP$ comonad is comonoidal as observed in~\cite[Theorem 17]{comonadicview}. Roughly speaking this boils down to the properties of the transformations $\MP(S\otimes T)\to \MP(S)\otimes \MP(T)$ which interpret measurement protocols in $S$ or in $T$ as measurement protocols over $S\otimes T$.}.

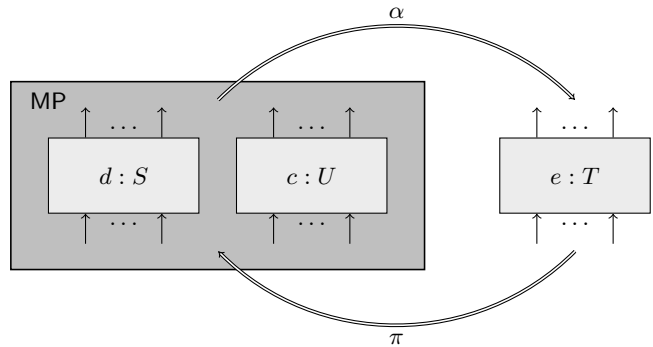
\begin{figure}
\begin{tikzpicture}
\draw[black, fill = gray, fill opacity = 0.5, semithick] (-3.5,-1.25) rectangle (2,1.25);
\node at (-3,1) {$\MP$};
\node[box] (a) [label=above:$\ldots$, label=below:$\ldots$] at (-2,0){$d:S$};

\node[box] (c) [label=above:$\ldots$, label=below:$\ldots$] at (0.5,0){$c:U$};

\node[box] (b) [label=above:$ \ldots $, label=below:$\ldots$] at (4,0) {$e:T$};

\draw [->] ([xshift=.5cm]a.north west) to +(0,.4);
\draw [->] ([xshift=-.5cm]a.north east) to +(0,.4);
\draw [<-] ([xshift=.5cm]a.south west) to +(0,-.4);
\draw [<-] ([xshift=-.5cm]a.south east) to +(0,-.4);

\draw [->] ([xshift=.5cm]c.north west) to +(0,.4);
\draw [->] ([xshift=-.5cm]c.north east) to +(0,.4);
\draw [<-] ([xshift=.5cm]c.south west) to +(0,-.4);
\draw [<-] ([xshift=-.5cm]c.south east) to +(0,-.4);

\draw [->] ([xshift=.5cm]b.north west) to +(0,.4);
\draw [->] ([xshift=-.5cm]b.north east) to +(0,.4);
\draw [<-] ([xshift=.5cm]b.south west) to +(0,-.4);
\draw [<-] ([xshift=-.5cm]b.south east) to +(0,-.4);

\draw [->, double,>=implies] ([yshift=.5cm,xshift=.25cm]a.north east) to[out=45,in=135]  node[above] {$\alpha$} ([yshift=.5cm]b.north);
\draw [->, double, >=implies] ([yshift=-.5cm]b.south) to[out=-135,in=-45] node[below] {$\pi$}  ([yshift=-.5cm,xshift=.25cm]a.south east);
\end{tikzpicture}\caption{General simulation $d\to e$, where we require $c$ to be noncontextual.}\label{fig:MPandrandomness}
\end{figure}

We will denote the existence of such a simulation $d\to e$ by $d\simulates e$, read as ``$d$ simulates $e$.''  Simulations thus defined interact well with contextuality. For instance, \cite[Theorem 21]{comonadicview} states that the noncontextual fraction, studied in~\cite{abramsky2017contextual}, is a monotone, that is, if $d\simulates e$, then $\NCF(d)\leq\NCF(e)$; and \cite[Theorem 4.1]{karvonen2018categories} implies that an empirical model is noncontextual if and only if it can be simulated from the trivial model on the empty scenario. In fact, the notions of logical and strong contextuality can be captured along similar lines~\cite{closingbell} by relaxing the equality of probability distributions in the definition of simulation by equality (or inclusion) of supports of these distributions.

We now state our main result for this resource theory of contextuality.
\begin{theorem}\label{thm:nocatalysisforcontextuality} If $d\otimes e\simulates d\otimes f$ in the resource theory of contextuality, then $e\simulates f$.
\end{theorem}
The key ideas of the proof are simple, even if the full details get technical: if $d$ can catalyze a transformation $e\to f$ once, it can do so arbitrarily many times. Choosing a big enough number of copies of $e$ to transform to copies of $f$, using the pigeonhole principle we can show that one only needs a compatible subset of $S_d$ (or rather, a compatible set of measurement protocols). Making this precise requires formalizing our framework more carefully, and we do this in the Supplemental Material \footnote{See Supplemental Material at \url{http://link.aps.org/supplemental/10.1103/PhysRevLett.127.160402} or at the end of this document for formal development of background material and full proofs.}.
 Our result subsumes~\cite[Theorem 22]{comonadicview}, as setting $d=f$ and letting $e$ be the trivial model on the empty scenario implies that if $d\simulates d\otimes d$, the model $d$ must be noncontextual.

\prlsection{No-catalysis for nonlocality}We know explain how to interpret nonlocality within this framework. At the level of scenarios and models on them, nonlocality can be seen as a special case of contextuality: for nonlocality, the measurement scenario typically arises by considering $n$ parties, with the $i$th party choosing one measurement from a set $X_i$ of measurements available to them, with a measurement $x\in X_i$ giving outcomes in some outcome set $O_{i,x}$. Often one restricts the situation even further and assumes that each party has the same number of measurements available and each measurement takes outcomes in a set of the same size, so that the scenario is specified (up to isomorphism) by three numbers: the number of parties, the number of measurements available to each of them and the size of the outcome sets. Whether or not one imposes this further restriction, such scenarios are of the form $S=\bigotimes_{i=1}^n S_i$ where each $S_i$ is just a discrete set of measurements (so only singleton measurements are possible in each $S_i$). In particular, a maximal measurement corresponds to a choice $(x_1,\dots x_n)$ of a measurement at each site, and a set of correlations can be given as a family $p(o_1,\dots o_n|x_1,\dots x_n)$ of conditional probabilities for each such measurement. If the family $p$ is (fully) no-signaling, it corresponds to a unique empirical model $e: S$ (where the probabilities over nonmaximal measurements are obtained by marginalization), and $p$ is local if and only if $e$ is noncontextual. 

However, if we allow parties to share different non-local resources, we move away from the situation where each party chooses one measurement from a set of mutually exclusive measurements. For instance, if Bob shares one box with Alice and one with Charlie, Bob is not limited to a single measurement: he can choose one for each box. As quantum theory allows for arbitrary joint measurability graphs in the case of projection-valued measurements~\cite{jointmeasurabilitygraphs} and arbitrary simplicial complexes in the case of positive-operator valued measurements~\cite{jointmeasurabilitystructures}, we impose no restrictions on the measurement scenarios available at each site. However, if we are thinking about a $n$-partite scenario, it is reasonable to expect that measurement choices done at one site do not affect the  measurements available at another one. Thus, we model an $n$-partite scenario as a tuple $(S_i)_{i=1}^n$ of scenarios, thought of as representing the $n$ parties sharing the scenario  $\bigotimes_{i=1}^n S_i$.

At the level of transformations between models and scenarios, nonlocality is no longer a special case of contextuality, as observed for example in~\cite[Appendix A.1]{Wolfeetal:quantifyingbell}. This is because our wirings are slightly too general as they allow the $i$th party to wire some of their measurements to measurements belonging to other parties, whereas operationally speaking it is reasonable to require each party to have access only to measurements available to them (and shared randomness).  Nevertheless, the resource theories are very closely related, as one can obtain the resource theory for $n$-partite nonlocality from that of contextuality via a general construction, used in the context of cryptography in~\cite{crypto}, that builds a resource theory of $n$-partite resources from a given resource theory. 

In our setting, this amounts to defining the resource theory of $n$-partite nonlocality as follows: an $n$-partite scenario is an $n$-tuple $(S_i)_{i=1}^n$ of scenarios (some of which may be empty), and an $n$-partite empirical model $e:(S_i)_{i=1}^n$ is an empirical model on $\bigotimes_{i=1}^n S_i$. The parallel composite $\boxtimes$ of scenarios is defined pointwise, i.e., by setting $(S_i)_{i=1}^n\boxtimes (T_i)_{i=1}^n=(S_i\otimes T_i)_{i=1}^n$. To define $\boxtimes$ for two $n$-partite models  $e:(S_i)_{i=1}^n$ and $d: (T_i)_{i=1}^n$, note that the scenarios $(\bigotimes_{i=1}^n S_i)\otimes (\bigotimes_{i=1}^n T_i)$ and $\bigotimes_{i=1}^n (S_i\otimes T_i)$ are canonically isomorphic so that the model $e\otimes d:(\bigotimes_{i=1}^n S_i)\otimes (\bigotimes_{i=1}^n T_i)$ induces a model $e\boxtimes d$ on $\bigotimes_{i=1}^n (S_i\otimes T_i)$ via this isomorphism. Finally, an $n$-partite simulation $d\to e$ is defined as an $n$-tuple of simulations $[\MP(T_i\otimes P_i\to S_i)]$ that, when taken together, transform $d\boxtimes c$ to $e$ where $c$ is some noncontextual shared correlation. In this manner, the resource theory of $n$-partite nonlocality is derived from that of contextuality by keeping track of the $n$-partite nature of scenarios and transformations between them.

Our proof of Theorem~\ref{thm:nocatalysisforcontextuality} readily implies that nonlocality admits no catalysts. 
\begin{theorem}\label{thm:nocatalysisfornonlocality}
 If $d\boxtimes e\simulates d\boxtimes f$ in the resource theory of $n$-partite nonlocality, then $e\simulates f$.
\end{theorem}
\prlsection{Quantum-assisted transformations and beyond}We briefly discuss a further generalization of our main result. One could consider resource theories with even more expressive simulations than the ones defined above. In the above, we define simulations as deterministic simulations assisted by noncontextual randomness. As suggested in~\cite{comonadicview}, one could allow more general correlations to be used in transformations and study, e.g., quantum-assisted simulations. More specifically, one could define quantum-assisted 
simulations $e\to f$ as deterministic simulations $\MP(e\otimes q)\to f$ where $q$ is a quantum-realizable empirical model. More generally, for any class $\X$ of empirical models that is closed under $\otimes$, one gets a well-defined notion of $\X$-assisted simulations, and we will write $d\simulates_\X e$ to denote the existence of an $\X$-assisted simulation from $d$ to $e$. If the class $\X$ contains all non-local correlations, the resulting resource theory has no catalysts. A similar result holds for $\X$-assisted $n$-partite transformations between $n$-partite correlations, but we restrict ourselves to stating the theorem for contextuality.
\begin{theorem}\label{thm:generalnocatalysis}
Let $\X$ be a class of empirical models that is closed under $\otimes$ and contains all classical correlations. Then  $d\otimes e\simulates_\X d\otimes f$ implies $e\simulates_\X f$.
\end{theorem}
For instance, if we take $\X$ to be the class of quantum-realizable empirical models, this result implies that, if by some miracle we got access to a single PR box, we could not use it to catalyze a quantum-assisted transformation that was hitherto impossible. Put another way, no matter what class of correlations we can use freely, the only way to get mileage out of a box that goes beyond our powers is to spend it---so we must choose wisely. 

\begin{acknowledgments}
We wish to thank Marcelo Terra Cunha for asking whether contextuality admits catalysts and for Ehtibar Dzhafarov for organizing the Quantum Contextuality in Quantum Mechanics and Beyond workshop where we met Marcelo and became interested in this question. We also thank Samson Abramsky, Rui Soares Barbosa and Shane Mansfield for helpful discussions.
\end{acknowledgments}

\bibliography{catalysisupdated}

\begin{thebibliography}{55}%
\makeatletter
\providecommand \@ifxundefined [1]{%
 \@ifx{#1\undefined}
}%
\providecommand \@ifnum [1]{%
 \ifnum #1\expandafter \@firstoftwo
 \else \expandafter \@secondoftwo
 \fi
}%
\providecommand \@ifx [1]{%
 \ifx #1\expandafter \@firstoftwo
 \else \expandafter \@secondoftwo
 \fi
}%
\providecommand \natexlab [1]{#1}%
\providecommand \enquote  [1]{``#1''}%
\providecommand \bibnamefont  [1]{#1}%
\providecommand \bibfnamefont [1]{#1}%
\providecommand \citenamefont [1]{#1}%
\providecommand \href@noop [0]{\@secondoftwo}%
\providecommand \href [0]{\begingroup \@sanitize@url \@href}%
\providecommand \@href[1]{\@@startlink{#1}\@@href}%
\providecommand \@@href[1]{\endgroup#1\@@endlink}%
\providecommand \@sanitize@url [0]{\catcode `\\12\catcode `\$12\catcode
  `\&12\catcode `\#12\catcode `\^12\catcode `\_12\catcode `\%12\relax}%
\providecommand \@@startlink[1]{}%
\providecommand \@@endlink[0]{}%
\providecommand \url  [0]{\begingroup\@sanitize@url \@url }%
\providecommand \@url [1]{\endgroup\@href {#1}{\urlprefix }}%
\providecommand \urlprefix  [0]{URL }%
\providecommand \Eprint [0]{\href }%
\providecommand \doibase [0]{https://doi.org/}%
\providecommand \selectlanguage [0]{\@gobble}%
\providecommand \bibinfo  [0]{\@secondoftwo}%
\providecommand \bibfield  [0]{\@secondoftwo}%
\providecommand \translation [1]{[#1]}%
\providecommand \BibitemOpen [0]{}%
\providecommand \bibitemStop [0]{}%
\providecommand \bibitemNoStop [0]{.\EOS\space}%
\providecommand \EOS [0]{\spacefactor3000\relax}%
\providecommand \BibitemShut  [1]{\csname bibitem#1\endcsname}%
\let\auto@bib@innerbib\@empty
%</preamble>
\bibitem [{\citenamefont {Kochen}\ and\ \citenamefont {Specker}(1967)}]{ks}%
  \BibitemOpen
  \bibfield  {author} {\bibinfo {author} {\bibfnamefont {S.}~\bibnamefont
  {Kochen}}\ and\ \bibinfo {author} {\bibfnamefont {E.~P.}\ \bibnamefont
  {Specker}},\ }\bibfield  {title} {\bibinfo {title} {The {P}roblem of {H}idden
  {V}ariables in {Q}uantum {M}echanics},\ }\href
  {https://www.jstor.org/stable/24902153} {\bibfield  {journal} {\bibinfo
  {journal} {J. Math. Mech.}\ }\textbf {\bibinfo {volume} {17}},\ \bibinfo
  {pages} {59} (\bibinfo {year} {1967})}\BibitemShut {NoStop}%
\bibitem [{\citenamefont {Budroni}\ \emph {et~al.}()\citenamefont {Budroni},
  \citenamefont {Cabello}, \citenamefont {G\"{u}hne}, \citenamefont
  {Kleinmann},\ and\ \citenamefont {Larsson}}]{Contextualityreview}%
  \BibitemOpen
  \bibfield  {author} {\bibinfo {author} {\bibfnamefont {C.}~\bibnamefont
  {Budroni}}, \bibinfo {author} {\bibfnamefont {A.}~\bibnamefont {Cabello}},
  \bibinfo {author} {\bibfnamefont {O.}~\bibnamefont {G\"{u}hne}}, \bibinfo
  {author} {\bibfnamefont {M.}~\bibnamefont {Kleinmann}},\ and\ \bibinfo
  {author} {\bibfnamefont {{\relax J.-Å.}.}~\bibnamefont {Larsson}},\
  }\href@noop {} {\bibinfo {title} {Quantum contextuality}},\ \Eprint
  {https://arxiv.org/abs/arXiv:2102.13036} {arXiv:2102.13036} \BibitemShut
  {NoStop}%
\bibitem [{\citenamefont {Bell}(1966)}]{bell1966}%
  \BibitemOpen
  \bibfield  {author} {\bibinfo {author} {\bibfnamefont {J.~S.}\ \bibnamefont
  {Bell}},\ }\bibfield  {title} {\bibinfo {title} {On the problem of hidden
  variables in quantum mechanics},\ }\href
  {https://doi.org/10.1103/RevModPhys.38.447} {\bibfield  {journal} {\bibinfo
  {journal} {Rev. Mod. Phys.}\ }\textbf {\bibinfo {volume} {38}},\ \bibinfo
  {pages} {447} (\bibinfo {year} {1966})}\BibitemShut {NoStop}%
\bibitem [{\citenamefont {Brunner}\ \emph {et~al.}(2014)\citenamefont
  {Brunner}, \citenamefont {Cavalcanti}, \citenamefont {Pironio}, \citenamefont
  {Scarani},\ and\ \citenamefont {Wehner}}]{bell-review}%
  \BibitemOpen
  \bibfield  {author} {\bibinfo {author} {\bibfnamefont {N.}~\bibnamefont
  {Brunner}}, \bibinfo {author} {\bibfnamefont {D.}~\bibnamefont {Cavalcanti}},
  \bibinfo {author} {\bibfnamefont {S.}~\bibnamefont {Pironio}}, \bibinfo
  {author} {\bibfnamefont {V.}~\bibnamefont {Scarani}},\ and\ \bibinfo {author}
  {\bibfnamefont {S.}~\bibnamefont {Wehner}},\ }\bibfield  {title} {\bibinfo
  {title} {Bell nonlocality},\ }\href
  {https://doi.org/10.1103/RevModPhys.86.419} {\bibfield  {journal} {\bibinfo
  {journal} {Rev. Mod. Phys.}\ }\textbf {\bibinfo {volume} {86}},\ \bibinfo
  {pages} {419} (\bibinfo {year} {2014})}\BibitemShut {NoStop}%
\bibitem [{\citenamefont {Scarani}\ \emph {et~al.}(2009)\citenamefont
  {Scarani}, \citenamefont {Bechmann-Pasquinucci}, \citenamefont {Cerf},
  \citenamefont {Du{\v{s}}ek}, \citenamefont {L{\"u}tkenhaus},\ and\
  \citenamefont {Peev}}]{scarani2009security}%
  \BibitemOpen
  \bibfield  {author} {\bibinfo {author} {\bibfnamefont {V.}~\bibnamefont
  {Scarani}}, \bibinfo {author} {\bibfnamefont {H.}~\bibnamefont
  {Bechmann-Pasquinucci}}, \bibinfo {author} {\bibfnamefont {N.~J.}\
  \bibnamefont {Cerf}}, \bibinfo {author} {\bibfnamefont {M.}~\bibnamefont
  {Du{\v{s}}ek}}, \bibinfo {author} {\bibfnamefont {N.}~\bibnamefont
  {L{\"u}tkenhaus}},\ and\ \bibinfo {author} {\bibfnamefont {M.}~\bibnamefont
  {Peev}},\ }\bibfield  {title} {\bibinfo {title} {The security of practical
  quantum key distribution},\ }\href
  {https://doi.org/10.1103/RevModPhys.81.1301} {\bibfield  {journal} {\bibinfo
  {journal} {Rev. Mod. Phys.}\ }\textbf {\bibinfo {volume} {81}},\ \bibinfo
  {pages} {1301} (\bibinfo {year} {2009})}\BibitemShut {NoStop}%
\bibitem [{\citenamefont {Ac{\'\i}n}\ and\ \citenamefont
  {Masanes}(2016)}]{acin2016certified}%
  \BibitemOpen
  \bibfield  {author} {\bibinfo {author} {\bibfnamefont {A.}~\bibnamefont
  {Ac{\'\i}n}}\ and\ \bibinfo {author} {\bibfnamefont {L.}~\bibnamefont
  {Masanes}},\ }\bibfield  {title} {\bibinfo {title} {Certified randomness in
  quantum physics},\ }\href {https://doi.org/10.1038/nature20119} {\bibfield
  {journal} {\bibinfo  {journal} {Nature}\ }\textbf {\bibinfo {volume} {540}},\
  \bibinfo {pages} {213} (\bibinfo {year} {2016})}\BibitemShut {NoStop}%
\bibitem [{\citenamefont {Fehr}\ \emph {et~al.}(2013)\citenamefont {Fehr},
  \citenamefont {Gelles},\ and\ \citenamefont {Schaffner}}]{fehr2013security}%
  \BibitemOpen
  \bibfield  {author} {\bibinfo {author} {\bibfnamefont {S.}~\bibnamefont
  {Fehr}}, \bibinfo {author} {\bibfnamefont {R.}~\bibnamefont {Gelles}},\ and\
  \bibinfo {author} {\bibfnamefont {C.}~\bibnamefont {Schaffner}},\ }\bibfield
  {title} {\bibinfo {title} {Security and composability of randomness expansion
  from {B}ell inequalities},\ }\href
  {https://doi.org/10.1103/PhysRevA.87.012335} {\bibfield  {journal} {\bibinfo
  {journal} {Phys. Rev. A}\ }\textbf {\bibinfo {volume} {87}},\ \bibinfo
  {pages} {012335} (\bibinfo {year} {2013})}\BibitemShut {NoStop}%
\bibitem [{\citenamefont {Anders}\ and\ \citenamefont
  {Browne}(2009)}]{Anders2009}%
  \BibitemOpen
  \bibfield  {author} {\bibinfo {author} {\bibfnamefont {J.}~\bibnamefont
  {Anders}}\ and\ \bibinfo {author} {\bibfnamefont {D.~E.}\ \bibnamefont
  {Browne}},\ }\bibfield  {title} {\bibinfo {title} {Computational {P}ower of
  {C}orrelations},\ }\href {https://doi.org/10.1103/PhysRevLett.102.050502}
  {\bibfield  {journal} {\bibinfo  {journal} {Phys. Rev. Lett.}\ }\textbf
  {\bibinfo {volume} {102}},\ \bibinfo {pages} {050502} (\bibinfo {year}
  {2009})}\BibitemShut {NoStop}%
\bibitem [{\citenamefont {Raussendorf}(2013)}]{raussendorf2013contextuality}%
  \BibitemOpen
  \bibfield  {author} {\bibinfo {author} {\bibfnamefont {R.}~\bibnamefont
  {Raussendorf}},\ }\bibfield  {title} {\bibinfo {title} {Contextuality in
  measurement-based quantum computation},\ }\href
  {https://doi.org/10.1103/PhysRevA.88.022322} {\bibfield  {journal} {\bibinfo
  {journal} {Phys. Rev. A}\ }\textbf {\bibinfo {volume} {88}},\ \bibinfo
  {pages} {022322} (\bibinfo {year} {2013})}\BibitemShut {NoStop}%
\bibitem [{\citenamefont {Howard}\ \emph {et~al.}(2014)\citenamefont {Howard},
  \citenamefont {Wallman}, \citenamefont {Veitch},\ and\ \citenamefont
  {Emerson}}]{howard2014contextuality}%
  \BibitemOpen
  \bibfield  {author} {\bibinfo {author} {\bibfnamefont {M.}~\bibnamefont
  {Howard}}, \bibinfo {author} {\bibfnamefont {J.}~\bibnamefont {Wallman}},
  \bibinfo {author} {\bibfnamefont {V.}~\bibnamefont {Veitch}},\ and\ \bibinfo
  {author} {\bibfnamefont {J.}~\bibnamefont {Emerson}},\ }\bibfield  {title}
  {\bibinfo {title} {Contextuality supplies the {`magic'} for quantum
  computation},\ }\href {https://doi.org/10.1038/nature13460} {\bibfield
  {journal} {\bibinfo  {journal} {Nature}\ }\textbf {\bibinfo {volume} {510}},\
  \bibinfo {pages} {351} (\bibinfo {year} {2014})}\BibitemShut {NoStop}%
\bibitem [{\citenamefont {Bermejo-Vega}\ \emph {et~al.}(2017)\citenamefont
  {Bermejo-Vega}, \citenamefont {Delfosse}, \citenamefont {Browne},
  \citenamefont {Okay},\ and\ \citenamefont
  {Raussendorf}}]{bermejo2017contextuality}%
  \BibitemOpen
  \bibfield  {author} {\bibinfo {author} {\bibfnamefont {J.}~\bibnamefont
  {Bermejo-Vega}}, \bibinfo {author} {\bibfnamefont {N.}~\bibnamefont
  {Delfosse}}, \bibinfo {author} {\bibfnamefont {D.~E.}\ \bibnamefont
  {Browne}}, \bibinfo {author} {\bibfnamefont {C.}~\bibnamefont {Okay}},\ and\
  \bibinfo {author} {\bibfnamefont {R.}~\bibnamefont {Raussendorf}},\
  }\bibfield  {title} {\bibinfo {title} {Contextuality as a resource for
  {M}odels of {Q}uantum {C}omputation with {Q}ubits},\ }\href
  {https://doi.org/10.1103/PhysRevLett.119.120505} {\bibfield  {journal}
  {\bibinfo  {journal} {Phys. Rev. Lett.}\ }\textbf {\bibinfo {volume} {119}},\
  \bibinfo {pages} {120505} (\bibinfo {year} {2017})}\BibitemShut {NoStop}%
\bibitem [{\citenamefont {Raussendorf}\ \emph {et~al.}(2017)\citenamefont
  {Raussendorf}, \citenamefont {Browne}, \citenamefont {Delfosse},
  \citenamefont {Okay},\ and\ \citenamefont
  {Bermejo-Vega}}]{raussendorf2017contextuality}%
  \BibitemOpen
  \bibfield  {author} {\bibinfo {author} {\bibfnamefont {R.}~\bibnamefont
  {Raussendorf}}, \bibinfo {author} {\bibfnamefont {D.~E.}\ \bibnamefont
  {Browne}}, \bibinfo {author} {\bibfnamefont {N.}~\bibnamefont {Delfosse}},
  \bibinfo {author} {\bibfnamefont {C.}~\bibnamefont {Okay}},\ and\ \bibinfo
  {author} {\bibfnamefont {J.}~\bibnamefont {Bermejo-Vega}},\ }\bibfield
  {title} {\bibinfo {title} {Contextuality and {W}igner-function negativity in
  qubit quantum computation},\ }\href
  {https://doi.org/10.1103/PhysRevA.95.052334} {\bibfield  {journal} {\bibinfo
  {journal} {Phys. Rev. A}\ }\textbf {\bibinfo {volume} {95}},\ \bibinfo
  {pages} {052334} (\bibinfo {year} {2017})}\BibitemShut {NoStop}%
\bibitem [{\citenamefont {Abramsky}\ \emph {et~al.}(2017)\citenamefont
  {Abramsky}, \citenamefont {Barbosa},\ and\ \citenamefont
  {Mansfield}}]{abramsky2017contextual}%
  \BibitemOpen
  \bibfield  {author} {\bibinfo {author} {\bibfnamefont {S.}~\bibnamefont
  {Abramsky}}, \bibinfo {author} {\bibfnamefont {R.~S.}\ \bibnamefont
  {Barbosa}},\ and\ \bibinfo {author} {\bibfnamefont {S.}~\bibnamefont
  {Mansfield}},\ }\bibfield  {title} {\bibinfo {title} {Contextual {F}raction
  as a {M}easure of {C}ontextuality},\ }\href
  {https://doi.org/10.1103/PhysRevLett.119.050504} {\bibfield  {journal}
  {\bibinfo  {journal} {Phys. Rev. Lett.}\ }\textbf {\bibinfo {volume} {119}},\
  \bibinfo {pages} {050504} (\bibinfo {year} {2017})}\BibitemShut {NoStop}%
\bibitem [{\citenamefont {Gao}\ \emph {et~al.}()\citenamefont {Gao},
  \citenamefont {Anschuetz}, \citenamefont {Wang}, \citenamefont {Cirac},\ and\
  \citenamefont {Lukin}}]{quantumML}%
  \BibitemOpen
  \bibfield  {author} {\bibinfo {author} {\bibfnamefont {X.}~\bibnamefont
  {Gao}}, \bibinfo {author} {\bibfnamefont {E.~R.}\ \bibnamefont {Anschuetz}},
  \bibinfo {author} {\bibfnamefont {S.-T.}\ \bibnamefont {Wang}}, \bibinfo
  {author} {\bibfnamefont {J.~I.}\ \bibnamefont {Cirac}},\ and\ \bibinfo
  {author} {\bibfnamefont {M.~D.}\ \bibnamefont {Lukin}},\ }\href@noop {}
  {\bibinfo {title} {Enhancing generative models via quantum correlations}},\
  \Eprint {https://arxiv.org/abs/arXiv:2101.08354} {arXiv:2101.08354}
  \BibitemShut {NoStop}%
\bibitem [{\citenamefont {Jonathan}\ and\ \citenamefont
  {Plenio}(1999)}]{catalysisofentanglement}%
  \BibitemOpen
  \bibfield  {author} {\bibinfo {author} {\bibfnamefont {D.}~\bibnamefont
  {Jonathan}}\ and\ \bibinfo {author} {\bibfnamefont {M.~B.}\ \bibnamefont
  {Plenio}},\ }\bibfield  {title} {\bibinfo {title} {Entanglement-{A}ssisted
  {L}ocal {M}anipulation of {P}ure {Q}uantum {S}tates},\ }\href
  {https://doi.org/10.1103/physrevlett.83.3566} {\bibfield  {journal} {\bibinfo
   {journal} {Phys. Rev. Lett.}\ }\textbf {\bibinfo {volume} {83}},\ \bibinfo
  {pages} {3566} (\bibinfo {year} {1999})}\BibitemShut {NoStop}%
\bibitem [{\citenamefont {M{\'{e}}thot}\ and\ \citenamefont
  {Scarani}(2007)}]{anomalyofnon-locality}%
  \BibitemOpen
  \bibfield  {author} {\bibinfo {author} {\bibfnamefont {A.~A.}\ \bibnamefont
  {M{\'{e}}thot}}\ and\ \bibinfo {author} {\bibfnamefont {V.}~\bibnamefont
  {Scarani}},\ }\bibfield  {title} {\bibinfo {title} {An anomaly of
  non-locality},\ }\href {https://doi.org/10.26421/QIC7.1-2-10} {\bibfield
  {journal} {\bibinfo  {journal} {Quantum Inf. Comput.}\ }\textbf {\bibinfo
  {volume} {7}},\ \bibinfo {pages} {157} (\bibinfo {year} {2007})}\BibitemShut
  {NoStop}%
\bibitem [{\citenamefont {Chitambar}\ and\ \citenamefont
  {Gour}(2019)}]{chitambar2019resource}%
  \BibitemOpen
  \bibfield  {author} {\bibinfo {author} {\bibfnamefont {E.}~\bibnamefont
  {Chitambar}}\ and\ \bibinfo {author} {\bibfnamefont {G.}~\bibnamefont
  {Gour}},\ }\bibfield  {title} {\bibinfo {title} {Quantum resource theories},\
  }\href {https://doi.org/10.1103/RevModPhys.91.025001} {\bibfield  {journal}
  {\bibinfo  {journal} {Rev. Mod. Phys}\ }\textbf {\bibinfo {volume} {91}},\
  \bibinfo {pages} {025001} (\bibinfo {year} {2019})}\BibitemShut {NoStop}%
\bibitem [{\citenamefont {Coecke}\ \emph {et~al.}(2016)\citenamefont {Coecke},
  \citenamefont {Fritz},\ and\ \citenamefont
  {Spekkens}}]{coecke2016mathematical}%
  \BibitemOpen
  \bibfield  {author} {\bibinfo {author} {\bibfnamefont {B.}~\bibnamefont
  {Coecke}}, \bibinfo {author} {\bibfnamefont {T.}~\bibnamefont {Fritz}},\ and\
  \bibinfo {author} {\bibfnamefont {R.~W.}\ \bibnamefont {Spekkens}},\
  }\bibfield  {title} {\bibinfo {title} {A mathematical theory of resources},\
  }\href {https://doi.org/10.1016/j.ic.2016.02.008} {\bibfield  {journal}
  {\bibinfo  {journal} {Inf. Comp.}\ }\textbf {\bibinfo {volume} {250}},\
  \bibinfo {pages} {59} (\bibinfo {year} {2016})}\BibitemShut {NoStop}%
\bibitem [{\citenamefont {Fritz}(2017)}]{fritz2017resource}%
  \BibitemOpen
  \bibfield  {author} {\bibinfo {author} {\bibfnamefont {T.}~\bibnamefont
  {Fritz}},\ }\bibfield  {title} {\bibinfo {title} {Resource convertibility and
  ordered commutative monoids},\ }\href
  {https://doi.org/10.1017/S0960129515000444} {\bibfield  {journal} {\bibinfo
  {journal} {Math. Struct. Comput. Sci.}\ }\textbf {\bibinfo {volume} {27}},\
  \bibinfo {pages} {850} (\bibinfo {year} {2017})}\BibitemShut {NoStop}%
\bibitem [{\citenamefont {Popescu}\ and\ \citenamefont
  {Rohrlich}(1994)}]{popescu1994quantum}%
  \BibitemOpen
  \bibfield  {author} {\bibinfo {author} {\bibfnamefont {S.}~\bibnamefont
  {Popescu}}\ and\ \bibinfo {author} {\bibfnamefont {D.}~\bibnamefont
  {Rohrlich}},\ }\bibfield  {title} {\bibinfo {title} {Quantum nonlocality as
  an axiom},\ }\href {https://doi.org/10.1007/BF02058098} {\bibfield  {journal}
  {\bibinfo  {journal} {Found. Phys.}\ }\textbf {\bibinfo {volume} {24}},\
  \bibinfo {pages} {379} (\bibinfo {year} {1994})}\BibitemShut {NoStop}%
\bibitem [{\citenamefont {Barrett}\ \emph {et~al.}(2005)\citenamefont
  {Barrett}, \citenamefont {Linden}, \citenamefont {Massar}, \citenamefont
  {Pironio}, \citenamefont {Popescu},\ and\ \citenamefont
  {Roberts}}]{barrettetal2005}%
  \BibitemOpen
  \bibfield  {author} {\bibinfo {author} {\bibfnamefont {J.}~\bibnamefont
  {Barrett}}, \bibinfo {author} {\bibfnamefont {N.}~\bibnamefont {Linden}},
  \bibinfo {author} {\bibfnamefont {S.}~\bibnamefont {Massar}}, \bibinfo
  {author} {\bibfnamefont {S.}~\bibnamefont {Pironio}}, \bibinfo {author}
  {\bibfnamefont {S.}~\bibnamefont {Popescu}},\ and\ \bibinfo {author}
  {\bibfnamefont {D.}~\bibnamefont {Roberts}},\ }\bibfield  {title} {\bibinfo
  {title} {Nonlocal correlations as an information-theoretic resource},\ }\href
  {https://doi.org/10.1103/PhysRevA.71.022101} {\bibfield  {journal} {\bibinfo
  {journal} {Phys. Rev. A}\ }\textbf {\bibinfo {volume} {71}},\ \bibinfo
  {pages} {022101} (\bibinfo {year} {2005})}\BibitemShut {NoStop}%
\bibitem [{\citenamefont {Barrett}\ and\ \citenamefont
  {Pironio}(2005)}]{barrettpironio2005}%
  \BibitemOpen
  \bibfield  {author} {\bibinfo {author} {\bibfnamefont {J.}~\bibnamefont
  {Barrett}}\ and\ \bibinfo {author} {\bibfnamefont {S.}~\bibnamefont
  {Pironio}},\ }\bibfield  {title} {\bibinfo {title} {{P}opescu-{R}ohrlich
  {C}orrelations as a {U}nit of {N}onlocality},\ }\href
  {https://doi.org/10.1103/PhysRevLett.95.140401} {\bibfield  {journal}
  {\bibinfo  {journal} {Phys. Rev. Lett.}\ }\textbf {\bibinfo {volume} {95}},\
  \bibinfo {pages} {140401} (\bibinfo {year} {2005})}\BibitemShut {NoStop}%
\bibitem [{\citenamefont {Allcock}\ \emph {et~al.}(2009)\citenamefont
  {Allcock}, \citenamefont {Brunner}, \citenamefont {Linden}, \citenamefont
  {Popescu}, \citenamefont {Skrzypczyk},\ and\ \citenamefont
  {V{\'e}rtesi}}]{allcock2009closedsets}%
  \BibitemOpen
  \bibfield  {author} {\bibinfo {author} {\bibfnamefont {J.}~\bibnamefont
  {Allcock}}, \bibinfo {author} {\bibfnamefont {N.}~\bibnamefont {Brunner}},
  \bibinfo {author} {\bibfnamefont {N.}~\bibnamefont {Linden}}, \bibinfo
  {author} {\bibfnamefont {S.}~\bibnamefont {Popescu}}, \bibinfo {author}
  {\bibfnamefont {P.}~\bibnamefont {Skrzypczyk}},\ and\ \bibinfo {author}
  {\bibfnamefont {T.}~\bibnamefont {V{\'e}rtesi}},\ }\bibfield  {title}
  {\bibinfo {title} {Closed sets of nonlocal correlations},\ }\href
  {https://doi.org/10.1103/PhysRevA.80.062107} {\bibfield  {journal} {\bibinfo
  {journal} {Phys. Rev. A}\ }\textbf {\bibinfo {volume} {80}},\ \bibinfo
  {pages} {062107} (\bibinfo {year} {2009})}\BibitemShut {NoStop}%
\bibitem [{\citenamefont {Jones}\ and\ \citenamefont
  {Masanes}(2005)}]{jonesmasanes2005interconversions}%
  \BibitemOpen
  \bibfield  {author} {\bibinfo {author} {\bibfnamefont {N.~S.}\ \bibnamefont
  {Jones}}\ and\ \bibinfo {author} {\bibfnamefont {L.}~\bibnamefont
  {Masanes}},\ }\bibfield  {title} {\bibinfo {title} {Interconversion of
  nonlocal correlations},\ }\href {https://doi.org/10.1103/PhysRevA.72.052312}
  {\bibfield  {journal} {\bibinfo  {journal} {Phys. Rev. A}\ }\textbf {\bibinfo
  {volume} {72}},\ \bibinfo {pages} {052312} (\bibinfo {year}
  {2005})}\BibitemShut {NoStop}%
\bibitem [{\citenamefont {Dupuis}\ \emph {et~al.}(2007)\citenamefont {Dupuis},
  \citenamefont {Gisin}, \citenamefont {Hasidim}, \citenamefont {M{\'e}thot},\
  and\ \citenamefont {Pilpel}}]{dupuis2007nouniversalbox}%
  \BibitemOpen
  \bibfield  {author} {\bibinfo {author} {\bibfnamefont {F.}~\bibnamefont
  {Dupuis}}, \bibinfo {author} {\bibfnamefont {N.}~\bibnamefont {Gisin}},
  \bibinfo {author} {\bibfnamefont {A.}~\bibnamefont {Hasidim}}, \bibinfo
  {author} {\bibfnamefont {A.~A.}\ \bibnamefont {M{\'e}thot}},\ and\ \bibinfo
  {author} {\bibfnamefont {H.}~\bibnamefont {Pilpel}},\ }\bibfield  {title}
  {\bibinfo {title} {No nonlocal box is universal},\ }\href
  {https://doi.org/10.1063/1.2767538} {\bibfield  {journal} {\bibinfo
  {journal} {J. Math. Phys.}\ }\textbf {\bibinfo {volume} {48}},\ \bibinfo
  {pages} {082107} (\bibinfo {year} {2007})}\BibitemShut {NoStop}%
\bibitem [{\citenamefont {De~Vicente}(2014)}]{vicente:nonlocality}%
  \BibitemOpen
  \bibfield  {author} {\bibinfo {author} {\bibfnamefont {J.~I.}\ \bibnamefont
  {De~Vicente}},\ }\bibfield  {title} {\bibinfo {title} {On nonlocality as a
  resource theory and nonlocality measures},\ }\href
  {https://doi.org/10.1088/1751-8113/47/42/424017} {\bibfield  {journal}
  {\bibinfo  {journal} {J. Phys. A}\ }\textbf {\bibinfo {volume} {47}},\
  \bibinfo {pages} {424017} (\bibinfo {year} {2014})}\BibitemShut {NoStop}%
\bibitem [{\citenamefont {Forster}\ and\ \citenamefont
  {Wolf}(2011)}]{forsterwolf2011bipartite}%
  \BibitemOpen
  \bibfield  {author} {\bibinfo {author} {\bibfnamefont {M.}~\bibnamefont
  {Forster}}\ and\ \bibinfo {author} {\bibfnamefont {S.}~\bibnamefont {Wolf}},\
  }\bibfield  {title} {\bibinfo {title} {Bipartite units of nonlocality},\
  }\href {https://doi.org/10.1103/PhysRevA.84.042112} {\bibfield  {journal}
  {\bibinfo  {journal} {Phys. Rev. A}\ }\textbf {\bibinfo {volume} {84}},\
  \bibinfo {pages} {042112} (\bibinfo {year} {2011})}\BibitemShut {NoStop}%
\bibitem [{\citenamefont {Gallego}\ \emph {et~al.}(2012)\citenamefont
  {Gallego}, \citenamefont {W{\"u}rflinger}, \citenamefont {Ac{\'\i}n},\ and\
  \citenamefont {Navascu{\'e}s}}]{gallego2012operational}%
  \BibitemOpen
  \bibfield  {author} {\bibinfo {author} {\bibfnamefont {R.}~\bibnamefont
  {Gallego}}, \bibinfo {author} {\bibfnamefont {L.~E.}\ \bibnamefont
  {W{\"u}rflinger}}, \bibinfo {author} {\bibfnamefont {A.}~\bibnamefont
  {Ac{\'\i}n}},\ and\ \bibinfo {author} {\bibfnamefont {M.}~\bibnamefont
  {Navascu{\'e}s}},\ }\bibfield  {title} {\bibinfo {title} {Operational
  {F}ramework for {N}onlocality},\ }\href
  {https://doi.org/10.1103/PhysRevLett.109.070401} {\bibfield  {journal}
  {\bibinfo  {journal} {Phys. Rev. Lett.}\ }\textbf {\bibinfo {volume} {109}},\
  \bibinfo {pages} {070401} (\bibinfo {year} {2012})}\BibitemShut {NoStop}%
\bibitem [{\citenamefont {Gallego}\ and\ \citenamefont
  {Aolita}(2017)}]{Gallego:wirings}%
  \BibitemOpen
  \bibfield  {author} {\bibinfo {author} {\bibfnamefont {R.}~\bibnamefont
  {Gallego}}\ and\ \bibinfo {author} {\bibfnamefont {L.}~\bibnamefont
  {Aolita}},\ }\bibfield  {title} {\bibinfo {title} {Nonlocality free wirings
  and the distinguishability between {B}ell boxes},\ }\href
  {https://doi.org/10.1103/PhysRevA.95.032118} {\bibfield  {journal} {\bibinfo
  {journal} {Phys. Rev. A}\ }\textbf {\bibinfo {volume} {95}},\ \bibinfo
  {pages} {032118} (\bibinfo {year} {2017})}\BibitemShut {NoStop}%
\bibitem [{\citenamefont {Wolfe}\ \emph {et~al.}(2020)\citenamefont {Wolfe},
  \citenamefont {Schmid}, \citenamefont {Sainz}, \citenamefont {Kunjwal},\ and\
  \citenamefont {Spekkens}}]{Wolfeetal:quantifyingbell}%
  \BibitemOpen
  \bibfield  {author} {\bibinfo {author} {\bibfnamefont {E.}~\bibnamefont
  {Wolfe}}, \bibinfo {author} {\bibfnamefont {D.}~\bibnamefont {Schmid}},
  \bibinfo {author} {\bibfnamefont {A.~B.}\ \bibnamefont {Sainz}}, \bibinfo
  {author} {\bibfnamefont {R.}~\bibnamefont {Kunjwal}},\ and\ \bibinfo {author}
  {\bibfnamefont {R.~W.}\ \bibnamefont {Spekkens}},\ }\bibfield  {title}
  {\bibinfo {title} {Quantifying {B}ell: the {R}esource {T}heory of
  {N}onclassicality of {C}ommon-{C}ause {B}oxes},\ }\href
  {https://doi.org/10.22331/q-2020-06-08-280} {\bibfield  {journal} {\bibinfo
  {journal} {{Quantum}}\ }\textbf {\bibinfo {volume} {4}},\ \bibinfo {pages}
  {280} (\bibinfo {year} {2020})}\BibitemShut {NoStop}%
\bibitem [{\citenamefont {Abramsky}\ and\ \citenamefont
  {Brandenburger}(2011)}]{ab}%
  \BibitemOpen
  \bibfield  {author} {\bibinfo {author} {\bibfnamefont {S.}~\bibnamefont
  {Abramsky}}\ and\ \bibinfo {author} {\bibfnamefont {A.}~\bibnamefont
  {Brandenburger}},\ }\bibfield  {title} {\bibinfo {title} {The sheaf-theoretic
  structure of non-locality and contextuality},\ }\href
  {https://doi.org/10.1088/1367-2630/13/11/113036} {\bibfield  {journal}
  {\bibinfo  {journal} {New J. Phys.}\ }\textbf {\bibinfo {volume} {13}},\
  \bibinfo {pages} {113036} (\bibinfo {year} {2011})}\BibitemShut {NoStop}%
\bibitem [{\citenamefont {Abramsky}\ \emph {et~al.}(2019)\citenamefont
  {Abramsky}, \citenamefont {Barbosa}, \citenamefont {Karvonen},\ and\
  \citenamefont {Mansfield}}]{comonadicview}%
  \BibitemOpen
  \bibfield  {author} {\bibinfo {author} {\bibfnamefont {S.}~\bibnamefont
  {Abramsky}}, \bibinfo {author} {\bibfnamefont {R.~S.}\ \bibnamefont
  {Barbosa}}, \bibinfo {author} {\bibfnamefont {M.}~\bibnamefont {Karvonen}},\
  and\ \bibinfo {author} {\bibfnamefont {S.}~\bibnamefont {Mansfield}},\
  }\bibfield  {title} {\bibinfo {title} {A comonadic view of simulation and
  quantum resources},\ }in\ \href {https://doi.org/10.1109/LICS.2019.8785677}
  {\emph {\bibinfo {booktitle} {Proceedings of 34th Annual ACM/IEEE Symposium
  on Logic in Computer Science (LiCS 2019)}}}\ (\bibinfo {organization}
  {IEEE},\ \bibinfo {year} {2019})\ pp.\ \bibinfo {pages} {1--12}\BibitemShut
  {NoStop}%
\bibitem [{\citenamefont {Broadbent}\ and\ \citenamefont
  {Karvonen}()}]{crypto}%
  \BibitemOpen
  \bibfield  {author} {\bibinfo {author} {\bibfnamefont {A.}~\bibnamefont
  {Broadbent}}\ and\ \bibinfo {author} {\bibfnamefont {M.}~\bibnamefont
  {Karvonen}},\ }\href@noop {} {\bibinfo {title} {Categorical composable
  cryptography}},\ \Eprint {https://arxiv.org/abs/arXiv:2105.05949}
  {arXiv:2105.05949} \BibitemShut {NoStop}%
\bibitem [{Note1()}]{Note1}%
  \BibitemOpen
  \bibinfo {note} {This is in contrast to asymptotic questions~\cite
  [V.B]{chitambar2019resource} or distillation tasks~\cite
  {Brunner:distillation} where one tries to approximate a target resource using
  increasing numbers of copies of the starting resource. It is unclear if the
  notion of a catalyst makes sense in such settings, which is why we work with
  single-shot convertibility.}\BibitemShut {Stop}%
\bibitem [{\citenamefont {Spekkens}(2005)}]{spekkens2005contextuality}%
  \BibitemOpen
  \bibfield  {author} {\bibinfo {author} {\bibfnamefont {R.~W.}\ \bibnamefont
  {Spekkens}},\ }\bibfield  {title} {\bibinfo {title} {Contextuality for
  preparations, transformations, and unsharp measurements},\ }\href
  {https://doi.org/10.1103/PhysRevA.71.052108} {\bibfield  {journal} {\bibinfo
  {journal} {Phys. Rev. A}\ }\textbf {\bibinfo {volume} {71}},\ \bibinfo
  {pages} {052108} (\bibinfo {year} {2005})}\BibitemShut {NoStop}%
\bibitem [{\citenamefont {Dzhafarov}\ and\ \citenamefont
  {Kujala}(2014)}]{ehtibar2014contextuality}%
  \BibitemOpen
  \bibfield  {author} {\bibinfo {author} {\bibfnamefont {E.~N.}\ \bibnamefont
  {Dzhafarov}}\ and\ \bibinfo {author} {\bibfnamefont {J.~V.}\ \bibnamefont
  {Kujala}},\ }\bibfield  {title} {\bibinfo {title} {Contextuality is about
  identity of random variables},\ }\href
  {https://doi.org/10.1088/0031-8949/2014/t163/014009} {\bibfield  {journal}
  {\bibinfo  {journal} {Phys. Scr.}\ }\textbf {\bibinfo {volume} {2014}},\
  \bibinfo {pages} {014009} (\bibinfo {year} {2014})}\BibitemShut {NoStop}%
\bibitem [{\citenamefont {Cabello}\ \emph {et~al.}(2014)\citenamefont
  {Cabello}, \citenamefont {Severini},\ and\ \citenamefont
  {Winter}}]{csw2014graphtheoretic}%
  \BibitemOpen
  \bibfield  {author} {\bibinfo {author} {\bibfnamefont {A.}~\bibnamefont
  {Cabello}}, \bibinfo {author} {\bibfnamefont {S.}~\bibnamefont {Severini}},\
  and\ \bibinfo {author} {\bibfnamefont {A.}~\bibnamefont {Winter}},\
  }\bibfield  {title} {\bibinfo {title} {Graph-{T}heoretic {A}pproach to
  {Q}uantum {C}orrelations},\ }\href
  {https://doi.org/10.1103/PhysRevLett.112.040401} {\bibfield  {journal}
  {\bibinfo  {journal} {Phys. Rev. Lett.}\ }\textbf {\bibinfo {volume} {112}},\
  \bibinfo {pages} {040401} (\bibinfo {year} {2014})}\BibitemShut {NoStop}%
\bibitem [{\citenamefont {Ac{\'\i}n}\ \emph {et~al.}(2015)\citenamefont
  {Ac{\'\i}n}, \citenamefont {Fritz}, \citenamefont {Leverrier},\ and\
  \citenamefont {Sainz}}]{acin2015combinatorial}%
  \BibitemOpen
  \bibfield  {author} {\bibinfo {author} {\bibfnamefont {A.}~\bibnamefont
  {Ac{\'\i}n}}, \bibinfo {author} {\bibfnamefont {T.}~\bibnamefont {Fritz}},
  \bibinfo {author} {\bibfnamefont {A.}~\bibnamefont {Leverrier}},\ and\
  \bibinfo {author} {\bibfnamefont {A.~B.}\ \bibnamefont {Sainz}},\ }\bibfield
  {title} {\bibinfo {title} {A combinatorial approach to nonlocality and
  contextuality},\ }\href {https://doi.org/10.1007/s00220-014-2260-1}
  {\bibfield  {journal} {\bibinfo  {journal} {Commun. Math. Phys.}\ }\textbf
  {\bibinfo {volume} {334}},\ \bibinfo {pages} {533} (\bibinfo {year}
  {2015})}\BibitemShut {NoStop}%
\bibitem [{\citenamefont {Staton}\ and\ \citenamefont
  {Uijlen}(2015)}]{sander2015effect}%
  \BibitemOpen
  \bibfield  {author} {\bibinfo {author} {\bibfnamefont {S.}~\bibnamefont
  {Staton}}\ and\ \bibinfo {author} {\bibfnamefont {S.}~\bibnamefont
  {Uijlen}},\ }\bibfield  {title} {\bibinfo {title} {Effect algebras,
  presheaves, non-locality and contextuality},\ }in\ \href
  {https://doi.org/10.1007/978-3-662-47666-6_32} {\emph {\bibinfo {booktitle}
  {Proceedings of 42nd International Colloquium on Automata, Languages, and
  Programming (ICALP 2015)}}},\ \bibinfo {series and number} {Lecture Notes in
  Computer Science},\ \bibinfo {editor} {edited by\ \bibinfo {editor}
  {\bibfnamefont {M.~M.}\ \bibnamefont {Halld{\'o}rsson}}, \bibinfo {editor}
  {\bibfnamefont {K.}~\bibnamefont {Iwama}}, \bibinfo {editor} {\bibfnamefont
  {N.}~\bibnamefont {Kobayashi}},\ and\ \bibinfo {editor} {\bibfnamefont
  {B.}~\bibnamefont {Speckmann}}}\ (\bibinfo  {publisher} {Springer},\ \bibinfo
  {year} {2015})\ pp.\ \bibinfo {pages} {401--413}\BibitemShut {NoStop}%
\bibitem [{\citenamefont {Joshi}\ \emph {et~al.}(2013)\citenamefont {Joshi},
  \citenamefont {Grudka}, \citenamefont {Horodecki}, \citenamefont {Horodecki},
  \citenamefont {Horodecki},\ and\ \citenamefont
  {Horodecki}}]{joshi2011nobroadcasting}%
  \BibitemOpen
  \bibfield  {author} {\bibinfo {author} {\bibfnamefont {P.}~\bibnamefont
  {Joshi}}, \bibinfo {author} {\bibfnamefont {A.}~\bibnamefont {Grudka}},
  \bibinfo {author} {\bibfnamefont {K.}~\bibnamefont {Horodecki}}, \bibinfo
  {author} {\bibfnamefont {M.}~\bibnamefont {Horodecki}}, \bibinfo {author}
  {\bibfnamefont {P.}~\bibnamefont {Horodecki}},\ and\ \bibinfo {author}
  {\bibfnamefont {R.}~\bibnamefont {Horodecki}},\ }\bibfield  {title} {\bibinfo
  {title} {No-broadcasting of non-signalling boxes via operations which
  transform local boxes into local ones},\ }\href
  {https://doi.org/10.26421/QIC13.7-8-2} {\bibfield  {journal} {\bibinfo
  {journal} {Quantum Inf. Comput.}\ }\textbf {\bibinfo {volume} {13}},\
  \bibinfo {pages} {0567} (\bibinfo {year} {2013})}\BibitemShut {NoStop}%
\bibitem [{\citenamefont {Bennett}\ \emph {et~al.}(1999)\citenamefont
  {Bennett}, \citenamefont {DiVincenzo}, \citenamefont {Fuchs}, \citenamefont
  {Mor}, \citenamefont {Rains}, \citenamefont {Shor}, \citenamefont {Smolin},\
  and\ \citenamefont {Wootters}}]{Bennetetal:loccvsaxiomatic}%
  \BibitemOpen
  \bibfield  {author} {\bibinfo {author} {\bibfnamefont {C.~H.}\ \bibnamefont
  {Bennett}}, \bibinfo {author} {\bibfnamefont {D.~P.}\ \bibnamefont
  {DiVincenzo}}, \bibinfo {author} {\bibfnamefont {C.~A.}\ \bibnamefont
  {Fuchs}}, \bibinfo {author} {\bibfnamefont {T.}~\bibnamefont {Mor}}, \bibinfo
  {author} {\bibfnamefont {E.}~\bibnamefont {Rains}}, \bibinfo {author}
  {\bibfnamefont {P.~W.}\ \bibnamefont {Shor}}, \bibinfo {author}
  {\bibfnamefont {J.~A.}\ \bibnamefont {Smolin}},\ and\ \bibinfo {author}
  {\bibfnamefont {W.~K.}\ \bibnamefont {Wootters}},\ }\bibfield  {title}
  {\bibinfo {title} {Quantum nonlocality without entanglement},\ }\href
  {https://doi.org/10.1103/PhysRevA.59.1070} {\bibfield  {journal} {\bibinfo
  {journal} {Phys. Rev. A}\ }\textbf {\bibinfo {volume} {59}},\ \bibinfo
  {pages} {1070} (\bibinfo {year} {1999})}\BibitemShut {NoStop}%
\bibitem [{\citenamefont {Heimendahl}\ \emph {et~al.}()\citenamefont
  {Heimendahl}, \citenamefont {Heinrich},\ and\ \citenamefont
  {Gross}}]{Heimendahl:axiomaitvsoperationalmagic}%
  \BibitemOpen
  \bibfield  {author} {\bibinfo {author} {\bibfnamefont {A.}~\bibnamefont
  {Heimendahl}}, \bibinfo {author} {\bibfnamefont {M.}~\bibnamefont
  {Heinrich}},\ and\ \bibinfo {author} {\bibfnamefont {D.}~\bibnamefont
  {Gross}},\ }\href@noop {} {\bibinfo {title} {The axiomatic and the
  operational approaches to resource theories of magic do not coincide}},\
  \Eprint {https://arxiv.org/abs/arXiv:2011.11651} {arXiv:2011.11651}
  \BibitemShut {NoStop}%
\bibitem [{\citenamefont {Barbosa}\ \emph {et~al.}()\citenamefont {Barbosa},
  \citenamefont {Karvonen},\ and\ \citenamefont {Mansfield}}]{closingbell}%
  \BibitemOpen
  \bibfield  {author} {\bibinfo {author} {\bibfnamefont {R.~S.}\ \bibnamefont
  {Barbosa}}, \bibinfo {author} {\bibfnamefont {M.}~\bibnamefont {Karvonen}},\
  and\ \bibinfo {author} {\bibfnamefont {S.}~\bibnamefont {Mansfield}},\
  }\href@noop {} {\bibinfo {title} {Closing {B}ell: {B}oxing black box
  simulations in the resource theory of contextuality}},\ \Eprint
  {https://arxiv.org/abs/arXiv:2104.11241} {arXiv:2104.11241} \BibitemShut
  {NoStop}%
\bibitem [{\citenamefont {Schmid}\ \emph {et~al.}()\citenamefont {Schmid},
  \citenamefont {Fraser}, \citenamefont {Kunjwal}, \citenamefont {Sainz},
  \citenamefont {Wolfe},\ and\ \citenamefont {Spekkens}}]{schmid:losr}%
  \BibitemOpen
  \bibfield  {author} {\bibinfo {author} {\bibfnamefont {D.}~\bibnamefont
  {Schmid}}, \bibinfo {author} {\bibfnamefont {T.~C.}\ \bibnamefont {Fraser}},
  \bibinfo {author} {\bibfnamefont {R.}~\bibnamefont {Kunjwal}}, \bibinfo
  {author} {\bibfnamefont {A.~B.}\ \bibnamefont {Sainz}}, \bibinfo {author}
  {\bibfnamefont {E.}~\bibnamefont {Wolfe}},\ and\ \bibinfo {author}
  {\bibfnamefont {R.~W.}\ \bibnamefont {Spekkens}},\ }\href@noop {} {\bibinfo
  {title} {Understanding the interplay of entanglement and nonlocality:
  {M}otivating and developing a new branch of entanglement theory}},\ \Eprint
  {https://arxiv.org/abs/arXiv:2004.09194} {arXiv:2004.09194} \BibitemShut
  {NoStop}%
\bibitem [{\citenamefont {Sengupta}\ \emph {et~al.}()\citenamefont {Sengupta},
  \citenamefont {Zibakhsh}, \citenamefont {Chitambar},\ and\ \citenamefont
  {Gour}}]{Senguptaetal:Nonlocalityisentanglement}%
  \BibitemOpen
  \bibfield  {author} {\bibinfo {author} {\bibfnamefont {K.}~\bibnamefont
  {Sengupta}}, \bibinfo {author} {\bibfnamefont {R.}~\bibnamefont {Zibakhsh}},
  \bibinfo {author} {\bibfnamefont {E.}~\bibnamefont {Chitambar}},\ and\
  \bibinfo {author} {\bibfnamefont {G.}~\bibnamefont {Gour}},\ }\href@noop {}
  {\bibinfo {title} {Quantum {B}ell nonlocality is entanglement}},\ \Eprint
  {https://arxiv.org/abs/arXiv:2012.06918} {arXiv:2012.06918} \BibitemShut
  {NoStop}%
\bibitem [{\citenamefont {Liang}\ \emph {et~al.}(2011)\citenamefont {Liang},
  \citenamefont {Spekkens},\ and\ \citenamefont {Wiseman}}]{Spekkerstriangle}%
  \BibitemOpen
  \bibfield  {author} {\bibinfo {author} {\bibfnamefont {Y.-C.}\ \bibnamefont
  {Liang}}, \bibinfo {author} {\bibfnamefont {R.~W.}\ \bibnamefont
  {Spekkens}},\ and\ \bibinfo {author} {\bibfnamefont {H.~M.}\ \bibnamefont
  {Wiseman}},\ }\bibfield  {title} {\bibinfo {title} {Specker's parable of the
  overprotective seer: A road to contextuality, nonlocality and
  complementarity},\ }\href {https://doi.org/10.1016/j.physrep.2011.05.001}
  {\bibfield  {journal} {\bibinfo  {journal} {Phys. Rep.}\ }\textbf {\bibinfo
  {volume} {506}},\ \bibinfo {pages} {1} (\bibinfo {year} {2011})}\BibitemShut
  {NoStop}%
\bibitem [{\citenamefont {Clauser}\ \emph {et~al.}(1969)\citenamefont
  {Clauser}, \citenamefont {Horne}, \citenamefont {Shimony},\ and\
  \citenamefont {Holt}}]{CHSH}%
  \BibitemOpen
  \bibfield  {author} {\bibinfo {author} {\bibfnamefont {J.~F.}\ \bibnamefont
  {Clauser}}, \bibinfo {author} {\bibfnamefont {M.~A.}\ \bibnamefont {Horne}},
  \bibinfo {author} {\bibfnamefont {A.}~\bibnamefont {Shimony}},\ and\ \bibinfo
  {author} {\bibfnamefont {R.~A.}\ \bibnamefont {Holt}},\ }\bibfield  {title}
  {\bibinfo {title} {Proposed {E}xperiment to {T}est {L}ocal
  {H}idden-{V}ariable {T}heories},\ }\href
  {https://doi.org/10.1103/physrevlett.23.880} {\bibfield  {journal} {\bibinfo
  {journal} {Phys. Rev. Lett.}\ }\textbf {\bibinfo {volume} {23}},\ \bibinfo
  {pages} {880} (\bibinfo {year} {1969})}\BibitemShut {NoStop}%
\bibitem [{\citenamefont {Amaral}\ \emph {et~al.}(2018)\citenamefont {Amaral},
  \citenamefont {Cabello}, \citenamefont {Cunha},\ and\ \citenamefont
  {Aolita}}]{amaral2017noncontextual}%
  \BibitemOpen
  \bibfield  {author} {\bibinfo {author} {\bibfnamefont {B.}~\bibnamefont
  {Amaral}}, \bibinfo {author} {\bibfnamefont {A.}~\bibnamefont {Cabello}},
  \bibinfo {author} {\bibfnamefont {M.~T.}\ \bibnamefont {Cunha}},\ and\
  \bibinfo {author} {\bibfnamefont {L.}~\bibnamefont {Aolita}},\ }\bibfield
  {title} {\bibinfo {title} {Noncontextual {W}irings},\ }\href
  {https://doi.org/10.1103/PhysRevLett.120.130403} {\bibfield  {journal}
  {\bibinfo  {journal} {Phys. Rev. Lett.}\ }\textbf {\bibinfo {volume} {120}},\
  \bibinfo {pages} {130403} (\bibinfo {year} {2018})}\BibitemShut {NoStop}%
\bibitem [{\citenamefont {Amaral}(2019)}]{amaral2019resource}%
  \BibitemOpen
  \bibfield  {author} {\bibinfo {author} {\bibfnamefont {B.}~\bibnamefont
  {Amaral}},\ }\bibfield  {title} {\bibinfo {title} {Resource theory of
  contextuality},\ }\href {https://doi.org/10.1098/rsta.2019.0010} {\bibfield
  {journal} {\bibinfo  {journal} {Philos. Trans. R. Soc. A}\ }\textbf {\bibinfo
  {volume} {377}},\ \bibinfo {pages} {20190010} (\bibinfo {year}
  {2019})}\BibitemShut {NoStop}%
\bibitem [{\citenamefont {Karvonen}(2019)}]{karvonen2018categories}%
  \BibitemOpen
  \bibfield  {author} {\bibinfo {author} {\bibfnamefont {M.}~\bibnamefont
  {Karvonen}},\ }\bibfield  {title} {\bibinfo {title} {Categories of empirical
  models},\ }in\ \href {https://doi.org/10.4204/EPTCS.287.14} {\emph {\bibinfo
  {booktitle} {Proceedings of 15th International Conference on Quantum Physics
  and Logic (QPL 2018)}}},\ \bibinfo {series} {Electronic Proceedings in
  Theoretical Computer Science}, Vol.\ \bibinfo {volume} {287},\ \bibinfo
  {editor} {edited by\ \bibinfo {editor} {\bibfnamefont {P.}~\bibnamefont
  {Selinger}}\ and\ \bibinfo {editor} {\bibfnamefont {G.}~\bibnamefont
  {Chiribella}}}\ (\bibinfo {year} {2019})\ pp.\ \bibinfo {pages}
  {239--252}\BibitemShut {NoStop}%
\bibitem [{Note2()}]{Note2}%
  \BibitemOpen
  \bibinfo {note} {The abstract explanation is that the $\protect \mathsf {MP}$
  comonad is comonoidal as observed in~\cite [Theorem 17]{comonadicview}.
  Roughly speaking this boils down to the properties of the transformations
  $\protect \mathsf {MP}(S\otimes T)\to \protect \mathsf {MP}(S)\otimes
  \protect \mathsf {MP}(T)$ which interpret measurement protocols in $S$ or in
  $T$ as measurement protocols over $S\otimes T$.}\BibitemShut {Stop}%
\bibitem [{Note3()}]{Note3}%
  \BibitemOpen
  \bibinfo {note} {See Supplemental Material at \protect \url
  {http://link.aps.org/supplemental/10.1103/PhysRevLett.127.160402} or at the
  end of this document for formal development of background material and full
  proofs.}\BibitemShut {Stop}%
\bibitem [{\citenamefont {Heunen}\ \emph {et~al.}(2014)\citenamefont {Heunen},
  \citenamefont {Fritz},\ and\ \citenamefont
  {Reyes}}]{jointmeasurabilitygraphs}%
  \BibitemOpen
  \bibfield  {author} {\bibinfo {author} {\bibfnamefont {C.}~\bibnamefont
  {Heunen}}, \bibinfo {author} {\bibfnamefont {T.}~\bibnamefont {Fritz}},\ and\
  \bibinfo {author} {\bibfnamefont {M.~L.}\ \bibnamefont {Reyes}},\ }\bibfield
  {title} {\bibinfo {title} {Quantum theory realizes all joint measurability
  graphs},\ }\href {https://doi.org/10.1103/physreva.89.032121} {\bibfield
  {journal} {\bibinfo  {journal} {Phys. Rev. A}\ }\textbf {\bibinfo {volume}
  {89}},\ \bibinfo {pages} {032121} (\bibinfo {year} {2014})}\BibitemShut
  {NoStop}%
\bibitem [{\citenamefont {Kunjwal}\ \emph {et~al.}(2014)\citenamefont
  {Kunjwal}, \citenamefont {Heunen},\ and\ \citenamefont
  {Fritz}}]{jointmeasurabilitystructures}%
  \BibitemOpen
  \bibfield  {author} {\bibinfo {author} {\bibfnamefont {R.}~\bibnamefont
  {Kunjwal}}, \bibinfo {author} {\bibfnamefont {C.}~\bibnamefont {Heunen}},\
  and\ \bibinfo {author} {\bibfnamefont {T.}~\bibnamefont {Fritz}},\ }\bibfield
   {title} {\bibinfo {title} {Quantum realization of arbitrary joint
  measurability structures},\ }\href
  {https://doi.org/10.1103/physreva.89.052126} {\bibfield  {journal} {\bibinfo
  {journal} {Phys. Rev. A}\ }\textbf {\bibinfo {volume} {89}},\ \bibinfo
  {pages} {052126} (\bibinfo {year} {2014})}\BibitemShut {NoStop}%
\bibitem [{\citenamefont {Brunner}\ and\ \citenamefont
  {Skrzypczyk}(2009)}]{Brunner:distillation}%
  \BibitemOpen
  \bibfield  {author} {\bibinfo {author} {\bibfnamefont {N.}~\bibnamefont
  {Brunner}}\ and\ \bibinfo {author} {\bibfnamefont {P.}~\bibnamefont
  {Skrzypczyk}},\ }\bibfield  {title} {\bibinfo {title} {Nonlocality
  {D}istillation and {P}ostquantum {T}heories with {T}rivial {C}ommunication
  {C}omplexity},\ }\href {https://doi.org/10.1103/physrevlett.102.160403}
  {\bibfield  {journal} {\bibinfo  {journal} {Phys. Rev. Lett.}\ }\textbf
  {\bibinfo {volume} {102}},\ \bibinfo {pages} {160403} (\bibinfo {year}
  {2009})}\BibitemShut {NoStop}%
\end{thebibliography}%

\newpage~\newpage

\title{Neither Contextuality nor Nonlocality Admits Catalysts---Supplemental Material}
\affiliation{University of Ottawa, Canada}

\maketitle

\appendix

In this Supplemental Material we provide proofs of our main theorems and some technical details concerning the background that is needed in order to formalize these.

\section{Background}\label{app:background}

We mostly follow the development of~\cite{comonadicview} with some changes in notation. For a scenario $S$ and subset $Y\subset{X_S}$ of measurements, we let $\Ev_S(Y)$ (or $\Ev(Y)$ when $S$ is clear from context) denote the set of possible joint outcomes for the (not necessarily comeasurable) set $Y$. Formally, $\Ev_S(Y)$ is defined as the cartesian product $\prod_{x\in Y}O_{S,x}$.

We now define empirical models more carefully. Given a probability distribution $d$ on $\Ev_S(Z)$, which we think of as a joint probability distribution for measurements in $Z$,  and a subset $Y\subset Z$, we denote by $d|_Y$ the marginalization of $d$ to measurements in $Y$. Now, an empirical model $e: S$ consists of a  family $(e_\sigma)_{\sigma\in\Sigma_S}$, where each $e_\sigma$ is a probability distribution over the set $\Ev_S(\sigma)$, and whenever $\tau\subset\sigma\in\Sigma_S$ we have
	\[e_\sigma|_\tau=e_\tau\]

The \emph{parallel composite} $S\otimes T$ of $S$ and $T$ is defined by $X_{S\otimes T}=X_S\sqcup X_T$, $\Sigma_{S\otimes T}=\{\sigma\sqcup \tau|\sigma\in\Sigma_S,\tau\in\Sigma_T\}$ and setting the outcome set at $x\in X_S\sqcup X_T$ to be $O_{S,x}$ if $x\in X_S$ and $O_{T,x}$ if $x\in X_T$. Here $X\sqcup Y$ denotes the disjoint union of the sets $X$ and $Y$, and can be defined \eg as $X\sqcup Y\defeq X\times\{0\}\cup Y\times\{1\}$.

A deterministic procedure $S\to T$ consists of a
	\begin{itemize}
		\item A simplicial function $\pi\colon X_T\to X_S$, i.e., a  function $\pi$ satisfying $\pi(\sigma)\in\Sigma_S$ for every $\sigma\in\Sigma_T$
		\item A family $\alpha=(\alpha_x\colon O_{S,\pi(x)}\to O_{T,x})_{x\in X}$ of functions between outcome sets.
	\end{itemize}
A procedure $\tuple{\pi,\alpha}\colon S\to T$ induces a mapping between empirical models so that each $e: S$ is pushed forward to $\tuple{\pi,\alpha}_*e: T$. Mathematically, the empirical model $\tuple{\pi,\alpha}_*e$ is defined  at measurement $\sigma\in \Sigma_T$ by
	\[(\tuple{\pi,\alpha}_*e)_\sigma=\alpha_*(e_{\pi\sigma})\]
where the right hand side denotes the pushforward of the probability distribution $e_{\pi\sigma}$ along the function $\Ev_S(\pi(\sigma))\to \Ev_T(\sigma)$ whose $x$th coordinate projection is given by $\Ev_S(\pi(\sigma))\to\Ev_S(\pi(x))\xrightarrow{\alpha_x}\Ev_T(x)$. 

If one imagines $S$ as corresponding to some particular experimental setup, and a model $e: S$ as corresponding to empirically observed propensities of outcomes for a fixed state preparation, then one can interpret a procedure $\tuple{\pi,\alpha}\colon S\to T$ as a way of (deterministically) building the experimental setup $T$ out of $S$: the map $\pi$ tells what measurement in $S$ each measurement in $T$ corresponds to, and $\alpha$ tells how to interpret outcomes in $S$ as outcomes in $T$. Then the model $\tuple{\pi,\alpha}_*e$ describes the statistics one would see if one was to observe the statistics given by $e$ but transform them according to $\tuple{\pi,\alpha}$. In other words, the probability $\tuple{\pi,\alpha}_*e$ gives to some fixed outcome over the joint measurement $\sigma\in\Sigma_T$ is the sum the probabilities $e$ gives to all outcomes of $\pi(\sigma)$ that are mapped to $s$ by $\alpha$.

A deterministic simulation $d\to e$ where $d:S$ and $e:T$ is a deterministic procedure $\tuple{\pi,\alpha}\colon S\to T$ that transforms $d$ to $e$, i.e., that satisfies $\tuple{\pi,\alpha}_*e=d$.

We now formalize measurement protocols carefully. These were used in~\cite[Appendix D]{acin2015combinatorial} to relate the sheaf-theoretic approach~\cite{ab} to contextuality to the hypergraph approach~\cite{acin2015combinatorial}, and then later used in~\cite{comonadicview} to extend the former. Intuitively, a (deterministic) measurement protocol is a set of rules that tells at each stage, what to measure next given the previous measurements and their outcomes. At any stage of the measurement protocol, we thus have a sequence $(x_i,o_i)_{i=1}^n$ of measurement-outcome pairs. A technical insight of~\cite{comonadicview} is that measurement protocols can be defined in terms of sets of such sequences: if you know all such sequences that could happen during a deterministic protocol, you also know the protocol itself. 

\begin{definition}
    A \emph{run} on a measurement scenario $S$ is a sequence $\x\defeq(x_i,o_i)_{i=1}^n$ such that $x_i\in X_S$ are distinct, $\enset{x_1,\ldots,x_n}\in\Sigma_S$, and each $o_i\in O_{S,x_i}$. We denote the empty run by $\Lambda$.
\end{definition}
    A run $\x$ determines a context
    $\sigma_{\x}\defeq\enset{x_1,\ldots, x_n} \in \Sigma_S$ and a joint assignment $s_{\x}\in \Ev(\sigma_{\x})$ on that context that maps $x_i$ to $o_i$. 
    Two runs $\x$ and $\y$ are said to be \emph{consistent} if they agree on common measurements, i.e.,  for every $z\in \sigma_{\x}\cap \sigma_{\y}$ we have $s_{\x}(z)=s_{\y}(z)$.
    
Given runs $\x$ and $\y$, we denote their concatenation by $\x\cdot\y$. Note that $\x\cdot\y$ might not be a run.

\begin{definition}
    A \emph{measurement protocol} on $S$ is a non-empty set $Q$ of runs satisfying the following  conditions:
    \begin{enumerate}[(i)]
        \item if $\x\cdot \y\in Q$ then $\x\in Q$;
        \item\label{cond:allpossibleoutcomes} if $\x\cdot (x,o)\in Q$, then $\x\cdot (x,o')\in Q$ for every $o'\in O_{S,x}$;
        \item\label{cond:determinacy} if $\x\cdot (x,o)\in Q$ and $\x\cdot (x',o')\in Q$, then $x=x'$. 
    \end{enumerate}
\end{definition}

Identifying a measurement protocol with its set of possible runs, the first condition guarantees that a prefix of a possible run is a possible run. The second condition ensures that, if $x$ might be measured at some stage, then any outcome of $x$ is an (in-principle) possible outcome at that stage, and the third condition guarantees that at any given stage, the next measurement prescribed by the protocol (if any), is uniquely determined.  

In this viewpoint, we may identify an outcome of a measurement protocol $Q$ with a \emph{maximal run} of $Q$, i.e., a run $\x\in Q$ that is not a proper prefix of any $\y\in Q$.  Such a run describes a sequence of measurement-outcome pairs obtained during $Q$, and maximality of the run ensures that the measurement protocol was followed to its conclusion. 

\begin{definition}
Given a scenario $S$, we build a scenario $\MP(S)$ as follows:
   \begin{itemize}
       \item its set of measurements is the set  $X_{\MP(S)}$ of measurement protocols on $S$;
       \item the outcome set $O_Q$ of a measurement protocol $Q \in X_{\MP(S)}$ is its set of maximal runs, i.e., those $\x\in Q$ that are not a proper prefix of any $\y\in Q$;
       \item a set $\enset{Q_1,\ldots Q_n}$ of measurement protocols is compatible
       whenever for any choice of pairwise consistent runs
       $\x_i \in Q_i$ with $i \in \enset{1,\ldots,n}$,
       we have $\bigcup_i \sigma_{\x_i} \in \Sigma$.
   \end{itemize}

    Given an empirical model $e:S$, we define the empirical model $\MP(e) : \MP(S)$ as follows.
    For a compatible set $\sigma\defeq\enset{Q_1,\ldots, Q_n}$ of measurement protocols and an assignment $\fdef{s}{Q_i}{\x_i} \in \Ev(\sigma)$, we set
\[\MP(e)_\sigma(s)\defeq\begin{cases} e_{\bigcup_{i} \sigma_{\x_i}}(\cup_i s_{\x_i}) &\text{ if $\enset{\x_i}$ pairwise consistent}\\
    								0 &\text{ otherwise. }\end{cases}\]
\end{definition}

The intuition behind a set of measurement protocols being compatible is that they can be performed together, without ever having to perform a measurement not allowed by $\Sigma_S$. We next make this intuition precise. 

\begin{definition} We say that a measurement protocol $P$ \emph{contains implicitly} a protocol $Q$, if any outcome of $P$ determines the outcome of $Q$, and we denote this by $P\succeq Q$. Formally, $P\succeq Q$ if 
for any maximal run $\x\in P$ there is a (necessarily unique) maximal run $\y\in Q$ such that the assignment $s_{\y}$ determined by $\y$ is a restriction of $s_{\x}$. We say that a measurement protocol $P$ contains implicitly a set of protocols $\{Q_1,\dots Q_n\}$ if $P\succeq Q_i$ for each $i$. 
\end{definition}

\begin{lemma}\label{lem:combiningMPs}
If $\{Q_1,\dots Q_n\}$ is a compatible set of measurement protocols, there is a measurement protocol $P$ that implicitly contains them all.
\end{lemma}

\begin{proof}
We define a protocol $P$ that first performs $Q_1$, then (whatever is left of) $Q_2$ and so on. Formally, we first define inductively a merge operation $*$ for compatible runs:
	\begin{align*}
	    \x*\Lambda &\defeq x
	    \\
	    \x*((y,o)\cdot \y) &\defeq \begin{cases} \x*\y &\text{ if }y\in \sigma_{\x} \\
										(\x\cdot (y,o))*\y &\text{ otherwise.}
						   \end{cases}
						   	\end{align*}
We extend $*$ to all pairs of runs by setting $\x*\y=\Lambda$ whenever $\x$ and $\y$ are not compatible. 

We then define $P$ by taking the closure of the set
\[\setdef{\x_1 * \dots * \x_n}{\x_i\in O_{\MP(S),Q_i}}\]
under prefixes. Now property (i) is true by construction, and properties (ii) and (iii) follow from each $Q_i$ being a measurement protocol Finally, each outcome of $P$ is of the form $\x_1 * \dots * \x_n$, so that $P\succeq Q_i$ for each $Q_i$.
\end{proof}

\begin{remark}\label{rem:productMPs}
If $P$ is a measurement protocol over $S$ and we fix the outcomes of some subset $Y\subset X_S$ of measurements to equal $t\in\Ev(Y)$, this determines a measurement protocol $P(t)$ that does not measure anything in $Y$ as follows: it proceeds exactly as $P$ except that whenever $P$ is supposed to measure some $x\in Y$ it behaves as if $t(x)$ was observed and proceeds accordingly. In particular, if $P$ is a measurement protocol over $S\otimes T$ and $t$ is an element of $\Ev_T(X_T)$, then $P(t)$ is a measurement protocol over $S$.

We can describe $P(t)$ more formally as follows. Let us say that a run $\x\in P$ is compatible with $t\in\Ev(Y)$ if $s_{\x}$ agrees with $t$ on common measurements, i.e., if $s_{\x}(z)=t(z)$ for every $z\in \sigma_{\x}\cap Y$. Given a run $\x$ and $t\in\Ev(Y)$, we define $\x\setminus t$ to be the run that omits measurements in $Y$. Formally, we can define this inductively on $\x$ by setting
	\begin{align*}
	    \Lambda\setminus t=\Lambda
	    \\
	    (\x\cdot(y,o))\setminus t &\defeq \begin{cases} \x\setminus t &\text{ if }y\in Y \\
										(\x\setminus t)\cdot (y,o)) &\text{ otherwise.}
						   \end{cases}
						   	\end{align*}
Then we can define $P(t)$ by
	\[P(t)\defeq\{\x\setminus t|\x\in P\text{ is compatible with }t\}\]
\end{remark}

\section{Proofs}\label{app:proofs}
\subsection{Proof of Theorem~\ref{thm:nocatalysisforcontextuality}}\label{app:contextuality}

Consider models $d:S$, $e:T$ and $f:U$ and an adaptive simulation $d\otimes e\to d\otimes f$. We wish to produce a simulation $e\to f$. Heuristically speaking, our overall strategy is to show that the composite $d\otimes e\to d\otimes f\to f$ needs to use $d$ only in a noncontextual manner. More precisely, we show that a compatible set of measurement protocols over $S$ suffices to carry out a simulation $d\otimes e\to f$. There are several compatible sets one could use for this: we choose one that works both for the resource theory of contextuality and for nonlocality.  

We first reduce to the deterministic case: by definition, there is a deterministic (but adaptive) map $d\otimes e\otimes c\to d\otimes f$ with $c$ noncontextual. Let us set $g\defeq e\otimes c$ and denote the scenario of $e\otimes c$ by $R$. We then have a deterministic map $d\otimes g\to d\otimes f$. Now, if $d$ can catalyze the transformation $g\to f$ once without getting spent in the process, it can do so arbitrarily many times by first catalyzing the first copy of $g$ to $f$, then the second one and so on. This results in a deterministic simulation $d\otimes g^{\otimes n}\to d\otimes f^{\otimes n}$ for any $n$ with the property that simulating the $i$th copy of $f$ uses only $d$ and the $i$th copy of $g$.

As the composite $d\otimes g^{\otimes n}\to d\otimes f^{\otimes n}\to f$ to the $i$th copy of $f$ uses only one copy of $g$ it hence induces a simulation $\tuple{\pi_i,\alpha_i}\colon d\otimes g\to f$. As there are only finitely many deterministic procedures $\MP(S\otimes R)\to U$, we can force as many of the simulations $\tuple{\pi_i,\alpha_i}$ to coincide by choosing large enough $n$. In particular, there is a single simulation that can work as many times as needed. More precisely, there is some fixed deterministic adaptive simulation $\tuple{\pi,\alpha}\colon d\otimes g\to f$, such that for any $n$ there is an deterministic adaptive simulation $d\otimes g^{\otimes n}\to d\otimes f^{\otimes n}$ with the property that simulating the  $i$th copy of $f$ uses only $d$ and the $i$th copy of $g$, but otherwise these simulations behave similarly, i.e., according to $\tuple{\pi,\alpha}$.

Consider now a fixed measurement $x$ and the measurement protocol $\pi(x)$ over $S\otimes R$ used to simulate it. Enumerate $\Ev_R(X_R)$ as $t_1,\dots t_k$. Since all the copies of of $x$ in $U^{\otimes n}$ are compatible with each other, so are their images under the simulation $d\otimes g^{\otimes n}\to f^{\otimes n}$. When simulating different copies of $x$ one queries a single copy of $d$ but different copies of $g$, which means that whatever results one gets in $g$ the measurements done in $d$ are always compatible. This implies that 
$\{\pi(x)(t_i)|i=1,\dots k\}$ is a compatible set of measurement protocols, where $\pi(x)(t_i)$ is defined as in Remark~\ref{rem:productMPs}. Thus they can be combined to a single measurement protocol $P_x$ by Lemma~\ref{lem:combiningMPs}. In particular, when simulating the $i$th copy of $x$ we can first measure $P_x$ obtaining an outcome $s$, and then perform $\pi(x)(s)$  (which is a measurement protocol over $R$) in the $i$th copy of $R$. We can then replace $d$ by the restriction $\hat{d}$ of $\MP(d)$ to $\MP(S)|_{\{P_x|x\in X_U\}}$. Moreover, considering distinct measurements $x$ belong to different copies of $f:U$, we see that this set has to be a compatible set of measurement protocols. As the set $\{P_x|x\in X_U\}$ is comeasurable, the model $\hat{d}$ is noncontextual, so that the deterministic simulation $\hat{d}\otimes g=\hat{d}\otimes e\otimes c\to f$ gives a probabilistic simulation $e\to f$ as desired. 
\newline 
\subsection{Proof of Theorem~\ref{thm:nocatalysisfornonlocality}}\label{app:nonlocality}

We now show how the above proof implies that there are no catalysts in the resource theory of nonlocality. Indeed, assume a simulation $(\pi_i,\alpha_i)_{i=1}^n\colon d\boxtimes e\to d\boxtimes f$ in the resource theory of nonlocality. Then $\bigotimes (\pi_i,\alpha_i)$ defines a simulation $d\boxtimes e\to d\boxtimes f$ in the resource theory of contextuality. Pre- and postcomposing by the canonical isomorphisms
\[ (\bigotimes_{i=1}^n S_i)\otimes (\bigotimes_{i=1}^n T_i)\cong\bigotimes_{i=1}^n (S_i\otimes T_i)\]
then implies that we have a simulation $d\otimes e\to d\otimes f$, so that the proof gives us a simulation $(\pi',\alpha')\colon e\to f$. However, we must check that this map is of the form $\bigotimes (\pi'_i,\alpha'_i)$ in order for it to give a simulation $e\to f$ in the resource theory of nonlocality. That this is true follows by inspecting the preceding proof. In fact, when simulating a measurement $x$ in the $T$  one performs a single measurement protocol in $S$ that captures everything one might need when measuring $x$, and then proceeds to $R$. If $x$ is a measurement at the $i$th measurement site, so is this ``everything'' and the protocol following it in $R$, so that the constructed simulation $e\to f$ is indeed of the required form.

\subsection{Proof of Theorem~\ref{thm:generalnocatalysis}}\label{app:general}

If $\X$ is closed under $\otimes$ and contains all noncontextual models, then $\X$-assisted simulation $d\otimes e\to d\otimes f$ is given by a deterministic simulation $d\otimes e\otimes c\to d\otimes f$ where $c\in \X$. We can then set $g\defeq e\otimes c$ to obtain a deterministic simulation $d\otimes g\to d\otimes f$ and retrace the proof of Theorem~\ref{thm:nocatalysisforcontextuality} to obtain a deterministic simulation $\hat{d}\otimes g=\hat{d}\otimes e\otimes c\to f$. As $\hat{d}$ is noncontextual, our assumptions on $\X$ imply that $\hat{d}\otimes c\in \X$, so that there is an $\X$-assisted simulation $e\to f$. 

The proof of Theorem~\ref{thm:nocatalysisfornonlocality} similarly gives a proof of no-catalysis for the resource theory of $\X$-assisted $n$-partite nonlocality.

\end{document}